\documentclass[hyper]{prop2015}
\usepackage[english]{babel}
\usepackage{amsmath,amssymb,amsfonts,amscd,amsthm,bbm}
\usepackage{eucal}
\usepackage{eufrak}
\usepackage{mathrsfs}
\usepackage{comment}
\usepackage{hyperref}  

\newtheorem{thm}{Theorem}
\newtheorem*{thmm}{Theorem}
\newtheorem{theorem}{Theorem}[section]

\newtheorem{proposition}{Proposition}[section]
\newtheorem{lemma}{Lemma}[section]

\newtheorem*{corollary}{Corollary}
\theoremstyle{definition}
\newtheorem{de}{Definition}[section]
\newtheorem{definition}{Definition}[section]
\newtheorem*{transflaw}{Transformation Law}
\newtheorem{ex}{Example}[section]
\newtheorem{example}{Example}[section]

\newtheorem*{remark}{Remark}

\newcommand{\co}{\colon\thinspace}
\newcommand{\mathcall}{\EuScript}

\renewcommand{\leq}{\leqslant}
\renewcommand{\geq}{\geqslant}

\renewcommand{\L}{{\Lambda}}

\DeclareMathOperator{\ad}{ad}

\DeclareMathOperator{\GL}{GL}
\DeclareMathOperator{\Ker}{Ker}
\renewcommand{\Im}{\mathop{\mathrm{Im}}}
\DeclareMathOperator{\Dim}{dim}
\DeclareMathOperator{\dimt}{\widetilde{dim}}
\DeclareMathOperator{\gh}{gh}

\DeclareMathOperator{\w}{{{w}}}  
\DeclareMathOperator{\pr}{{{\mathrm{\varepsilon}}}} 

\DeclareMathOperator{\Vect}{\mathrm{Vect}}
\DeclareMathOperator{\Vectn}{\Vect^{{\displaystyle{}-{}}}}
\DeclareMathOperator{\Vectp}{\Vect^{+}}

\DeclareMathOperator{\Mat}{Mat}
\DeclareMathOperator{\Hom}{Hom}

\DeclareMathOperator{\Fun}{\mathit{C}^{\infty}}

\DeclareMathOperator{\Map}{\mathrm{Map}}
\DeclareMathOperator{\Mapp}{\mathbf{Map}}
\DeclareMathOperator{\Fune}{\mathscr{E}}
\newcommand{\Funes}{\mathscr{E}}

\newcommand{\der}[2]{{\frac{\partial {#1}}{\partial {#2}}}}
\newcommand{\lder}[2]{{\partial {#1}/\partial {#2}}}
\newcommand{\dder}[3]{{\frac{\partial^2 {#1}}{\partial {#2}\partial {#3}}}}
\newcommand{\var}[2]{{\frac{\delta {#1}}{\delta {#2}}}}

\newcommand{\R}[1]{{\mathbbm R}^{#1}}
\newcommand{\RR}{\mathbbm R}
\newcommand{\Z}{{\mathbbm Z_{2}}}
\newcommand{\ZZ}{{\mathbbm Z}}
\newcommand{\ZP}{{\mathbbm Z}^{\geq 0}}
\newcommand{\SZ}{\hat{\mathbbm Z}}

\newcommand{\NN}{{\mathbbm N}}

\newcommand{\p}{\partial}
\newcommand{\fun}{C^{\infty}}

\newcommand{\frakg}{\mathfrak{g}}

\renewcommand{\a}{{\alpha}}
\renewcommand{\b}{{\beta}}
\newcommand{\F}{{\Phi}}

\newcommand{\h}{{\eta}}
\newcommand{\f}{{\varphi}}
\newcommand{\la}{{\lambda}}
\newcommand{\x}{{\xi}}

\newcommand{\e}{{\varepsilon}}

\newcommand{\s}{{\sigma}}

\newcommand{\xp}{{\boldsymbol{x}}}
\newcommand{\up}{{\boldsymbol{u}}}

\newcommand{\att}{{\tilde \alpha}}

\newcommand{\ft}{{\tilde f}}

\newcommand{\itt}{{\tilde i}}

\newcommand{\at}{{\tilde a}}
\newcommand{\bt}{{\tilde b}}

\newcommand{\ut}{{\tilde u}}
\newcommand{\vt}{{\tilde v}}
\newcommand{\xt}{{\tilde x}}

\newcommand{\funn}{\mathbf{C}^{\infty}}

\DeclareMathOperator{\funnh}{\mathbf{C^{\infty}_{\hbar}}}
\DeclareMathOperator{\pfunn}{\boldsymbol{\Pi}\mathbf{\!C^{\,\infty}}}
\newcommand{\tto}{{\linethickness{1.7pt}
		  \,\begin{picture}(1,0)
                   \put(0.075,0.296){\line(1,0){0.75}}
                   \put(0,0){$\boldsymbol{\Rightarrow}$}
                  \end{picture}
                  }\,
}

\newcommand{\oto}{{\linethickness{0.5pt}
		  \,\begin{picture}(1,0)
		  \put(0.07,0.175){\line(0,1){0.2}}
                   \put(-0.01,0){$\boldsymbol{\Rightarrow}$}
                  \end{picture}
                  }\,
}

\DeclareMathOperator{\ofun}{\mathit{OC_{\hbar}^{\infty}}}
\DeclareMathOperator{\SMan}{{\mathcall{SM}an}}

\newcommand{\ttoq}{\tto_{\hbar}}

\newcommand{\Sinf}{S_{\infty}}
\newcommand{\Pinf}{P_{\infty}}
\newcommand{\Linf}{L_{\infty}}

\newcommand{\Sinfh}{S_{\infty,\hbar}}

\newcommand{\Dbar}{{\mathchar '26 \mkern-11 mu D}}

\newcommand{\vp}{\boldsymbol{v}}
\newcommand{\infto}{\rightsquigarrow}

\DeclareMathOperator{\EThick}{{{\mathcall{ET}hick}}}
\DeclareMathOperator{\OThick}{{{\mathcall{OT}hick}}}

\newcommand{\kir}{\boldsymbol{\kappa}}
\newcommand{\okir}{\boldsymbol{\chi}}

\newcommand{\D}{\Delta}
\DeclareMathOperator{\ord}{ord}
\DeclareMathOperator{\graph}{graph}
\DeclareMathOperator{\sign}{sgn}

\newcommand{\0}{{00}}

\newcommand{\zeroloc}{\mathfrak{Z}}
\newcommand{\blocq}{\mathcal{B}}

\category{Proceedings}
\subtitle{\href{http://www.maths.dur.ac.uk/lms/109/index.html}{LMS/EPSRC Durham Symposium on Higher Structures in M-Theory}}
\title{Graded Geometry, $Q$-Manifolds, and Microformal Geometry}
\author[Th.~Th. Voronov]{Theodore~Th. Voronov\inst{a,b}\footnote{Corresponding author e-mail:~\href{mailto:theodore.voronov@gmail.com}{\textsf{theodore.voronov@gmail.com}}}}
\address[1]{School of Mathematics,  University of Manchester,    Manchester,   M13 9PL,  United Kingdom}
\address[2]{Faculty of Physics, Tomsk State University, Tomsk, 634050, Russia}
\begin{acknowledgements}
The author thanks the organizers of the LMS--EPSRC Durham Symposium ``Higher Structures in M-Theory'' (Durham, 12--18 August 2018) for the invitation and very stimulating atmosphere, and in particular for the suggestion to write this text.
\end{acknowledgements}
\begin{abstract}
We give an exposition of graded   and microformal geometry, and the language of $Q$-manifolds. $Q$-manifolds are supermanifolds endowed with an odd vector field of square zero. They can be seen as a non-linear analogue of Lie algebras (in parallel with even and odd Poisson manifolds),   a basis of ``non-linear homological algebra'', and    a powerful tool for describing algebraic and geometric structures. This language goes together with that of
graded manifolds, which are supermanifolds with an extra $\mathbbm{Z}$-grading in the structure sheaf. ``Microformal geometry'' is a new notion referring to   ``thick'' or ``microformal''  morphisms, which generalize ordinary smooth maps, but whose crucial feature is that the corresponding pullbacks of functions are   nonlinear. In particular,  ``Poisson thick morphisms'' of homotopy Poisson supermanifolds induce  $L_{\infty}$-morphisms of homotopy Poisson brackets. There is a quantum version based on special type Fourier integral operators and applicable to Batalin--Vilkovisky geometry.
Though the text is mainly expository, some results   are new or   not   published previously.
\end{abstract}
\shortabstract
\begin{document}
\maketitle

\section{Introduction}

The purpose of this text is to give an overview of graded geometry, i.e. the theory of graded manifolds, which are a version of supermanifolds (namely, supermanifolds endowed with an extra grading by integers in the algebra of functions) that have attracted much attention in recent years, and   an introduction to the new   area of microformal geometry, whose main feature is the new notion of ``microformal'' or ``thick'' morphisms generalizing ordinary smooth maps. These two topics are related by the type of applications, which are structures such as homotopy algebras   ultimately motivated by physics, in particular by ``gauge symmetries'' in broad sense. Key for description of homotopy structures is the language of $Q$-manifolds (see below), to which we  give a brief introduction as well.

``Thick morphisms'' (defined for ordinary manifolds, supermanifolds or graded manifolds) generalize ordinary maps or supermanifold morphisms, but are not maps themselves. They are defined as special type canonical relations or correspondences between the cotangent bundles. Canonical relations have long been a standard tool in symplectic or Poisson geometry, perceived as an extension of the notion of a canonical transformation (symplectomorphism) or a Poisson map. In the context of microformal geometry they play a different role as they are used for replacing ordinary maps of manifolds (the bases of the cotangent bundles). We define pullbacks of functions by thick morphisms, with the crucial new property of being (in general) non-linear. More precisely, the pullback by a thick morphism $\F\co M_1\tto M_2$ is a formal non-linear differential operator $\F^*\co \funn(M_2)\to \funn(M_1)$, which is a formal perturbation of an ordinary pullback (by some map that ``sits inside'' any thick morphism). Such a non-linear transformation has a remarkable feature that its derivative or variation at every function is the ordinary pullback by an ordinary map (more precisely, a formal perturbation of an ordinary map) and hence an algebra homomorphism. It remains an open question whether the non-linear pullbacks by thick morphisms can be characterized by this property.

The discovery of thick morphisms resulted from our search of a natural differential-geometric construction that would give non-linear maps of spaces of functions regarded as infinite-dimensional (super)manifolds. This was necessary for $\Linf$-morphisms of bracket structures.
Indeed, the most efficient way of describing various bracket structures, particularly  homotopy bracket structures, is the language of $Q$-manifolds, i.e. supermanifolds, possibly   graded manifolds, endowed with an odd vector field $Q$ satisfying $Q^2=0$.
The superiority of this geometric language  is proved when morphisms are considered: complicated and non-obvious algebraic definitions e.g. for $\Linf$-algebras or Lie algebroids over different bases are described with great simplification and uniformity as nothing but $Q$-maps of $Q$-manifolds, i.e. maps of the underlying supermanifolds that intertwine the corresponding vector fields $Q_1$ and $Q_2$. Non-linearity in such a map  is responsible for ``higher homotopies'' in the algebraic language.

(In physics parlance,  a homological vector field $Q$ is an infinitesimal ``BRST transformation''. In mathematics, particular instances of homological vector field have been known as various  differentials, e.g. the de Rham differential or Chevalley--Eilenberg differential. The power of the $Q$-manifold language was demonstrated by Kontsevich's formulation and proof of the formality theorem implying the existence of deformation quantization of arbitrary Poisson structures, which would be impossible without it.)

Therefore, in the case when a bracket structure is defined on functions, there is the need for a construction naturally leading to non-linear maps between spaces of functions. Clearly, ordinary pullbacks cannot serve this purpose as they are algebra homomorphisms and in particular linear. We came to the new ``non-linear pullbacks'' of functions and   the underlying ``thick  morphisms'' of (super)manifolds by solving a very concrete problem concerning the higher analogue of Koszul bracket on differential forms (introduced earlier by H.~Khudaverdian and the author) corresponding to a homotopy  Poisson structure. In the classical case of a usual Poisson bracket and the induced by it odd  Koszul bracket on forms, the classical fact in Poisson geometry was that raising indices with the help of the Poisson tensor maps the Koszul bracket on forms to the canonical Schouten bracket (the ``antibracket'') on multivectors. In the homotopy case, an analogue of that posed a big problem, since there is only one antibracket on multivectors and a whole infinite sequence of ``higher Koszul brackets'' on forms. Hence only an $\Linf$-morphism linking them would be possible, i.e. a non-linear transformation of forms to multivectors. This has been indeed achieved with the help of thick morphisms and pullbacks by them. It is absolutely certain now that any theory of homotopy brackets on manifolds (super or graded) should use thick morphisms and will be incomplete otherwise.

Although as explained microformal geometry, i.e. the theory of (super-, graded) manifolds with thick morphisms, owes its birth to Poisson and other bracket structures and their homotopy versions,  we wish to stress that in itself it does not assume on manifolds any additional structure and as such is an extension of differential topology with a larger class of morphisms. Also, the applications require considering graded or super case and this is the most natural framework for us, but the construction of thick morphisms  has nothing particularly super  as such and makes perfect sense in an entirely even context.

In the super case, there are two parallel versions of thick morphisms, adapted for pullback of even and odd functions respectively (``bosonic'' and ``fermionic''). Indeed, a non-linear transformation  cannot be applied indiscriminately to elements of an algebra satisfying different commutation rules, hence the need to distinguish between even and odd functions. While the ``bosonic'' version of thick morphism uses the symplectic geometry of cotangent bundles $T^*M$, the ``fermionic'' version uses the odd symplectic structure on anticotangent bundles $\Pi T^*M$. Also, the bosonic case can be further lifted on a quantum level. There are ``quantum thick morphisms'', which are (up to reversion of arrows) particular type Fourier integral operators. The ``classical'' thick morphisms are recovered in the limit $\hbar\to 0$ (similarly with Hamilton--Jacobi equation and Schr\"{o}dinger equation).

To put  the topics of this paper in a broader context, recall that there is a general philosophical principle of a certain ``duality'' between algebraic and geometric languages. More specifically, there is  a duality between commutative algebras and ``spaces'' (understood in the broadest sense). With every ``space'' (such as a topological space or a manifold or an algebraic variety) we can associate an algebra, which is an appropriate algebra of functions, and with a map of such spaces we can associate an algebra homomorphism in the opposite direction, given by the pullback. Conversely, every commutative algebra can be morally regarded as an algebra of functions and algebra homomorphisms as morally corresponding to maps of spaces, with the reversion of arrows. This heuristic principle can be traced back to the results of Stone and Kolmogorov--Gelfand in 1930s and Gelfand's duality between compact Hausdorff spaces and Banach algebras, and is fully realized in   Grothendieck's theory of schemes. Application of this principle of algebraic-geometric duality to graded algebras leads to supergeometry and the theory of graded manifolds considered in this paper. Applying it to differential graded algebras gives   $Q$-manifolds. Therefore the theory of the latter is a ``non-linear homological algebra''. If one further applies to it  the central idea of modern  homological algebra, i.e. that of derived category (considering complexes up to quasi-isomorphism), the outcome will be ``derived geometry'' (which we do not consider here; see e.g.~\cite{pridham:outline2018}). Now, the emergence of thick morphisms and non-linear pullbacks indicates at a non-linear extension of algebraic-geometric duality, which needs to be understood. Combination of our ``microformal geometry'' with homological and homotopical ideas seems to us  a fruitful  direction of possible future study.

Note that both graded geometry and microformal geometry were ultimately motivated by structures coming from physical problems e.g.  gauge theory. One of the purposes of this text is to be read by physicists and I tried to make the paper readable.

The structure of the paper is as follows. We begin from standard algebraic preliminaries and then pass to the definition and constructions of graded manifolds (Section~\ref{sec.gradedgeom}). For those familiar with the subject, we can say that we stress the distinction between grading responsible for signs ($\Z$-grading or parity) and $\ZZ$-grading, which we call ``weight'' (in physics it can be e.g. ghost number). They can be related, but do not have to, as serving different purposes. We also distinguish between   general $\ZZ$-graded case and the case of non-negatively graded manifolds. The latter have a natural structure of a fiber bundle with particular polynomial transformations as the structure group. They can be seen as a generalization of vector bundles. (Our general thesis is that a $\ZZ$-grading is a replacement of a missing linear structure.) We show that to every such non-linear fiber bundle there  corresponds canonically a graded vector bundle of a larger dimension containing all the information about the transition functions of the original bundle. (We call that ``canonical linearization''.) In Section~\ref{sec.qman}, we introduce $Q$-manifolds and explain how they can be used for describing various  structures, in particular homotopy bracket structures such as $\Linf$-algebras and $\Linf$-algebroids. This language is applied in the next two sections. In Section~\ref{sec.microclass}, we define thick morphisms of supermanifolds and give their main properties (the most important of which is the formula for the derivative of pullback). Then we show that a homotopy Poisson thick morphism induces an $\Linf$-morphism of the corresponding homotopy Poisson algebras. On the way we recall $\Sinf$- and $\Pinf$-structures. We also introduce a ``non-linear adjoint'' (an analogue of adjoint for a non-linear operator) as a thick morphism and give an application to $\Linf$-algebroids. In Section~\ref{sec.microquant}, we describe ``quantum microformal geometry''.  In particular, we show how quantum thick morphisms give $\Linf$-morphisms for ``quantum brackets'' generated by a higher-order BV operator.


There are many other things that we wanted to include in the paper, but could not because of time and space limitations. Also, we do not claim any completeness of the given bibliography, though we tried to provide accurate historic references.

A note about terminology. We often, but not always, drop the prefix `super-' (as well as the adjective `graded') and can speak (for example) about `manifolds' or `Lie algebras' meaning `supermanifolds' and `Lie superalgebras'. After we will have explained the `graded' notions, we will be assuming by default that grading can be included in all our constructions and will make that explicit   only when specifically need it.

\section{Graded geometry}
\label{sec.gradedgeom}

\subsection{Graded notions. Algebraic preliminaries}

\subsubsection{Basic definitions}
Recall some general algebraic notions. Let $G$ be an abelian group   written additively. A ``grading'' with values in $G$ is attaching labels $\la\in G$ to elements of some object; an element  to which such a label is attached is called ``homogeneous''. In such a generality, grading makes sense for sets. A set $S$ is \emph{$G$-graded}   if $S=\cup_{\la}  S_{\la}$ (disjoint union). Fix $G$ and in the future say ``graded'' for ``$G$-graded'' unless we need to clarify $G$. Presentation of $S$ in such a form is a ``graded structure''. More often the notion of grading is applied in the additive situation: to abelian groups, vector spaces and rings.
A vector space or a module $M$ over some ring     is called \emph{graded} if it has a form $M=\oplus M_{\la}$, where $M_{\la}\subset M$ are submodules. In the sequel we always assume that such a presentation is fixed as part of structure. Notation-wise, the index $\la$ can be written as a lower index or as an upper index depending on convenience and the typical convention is that $M_{\la}=M^{-\la}$. Some sources promote the idea of defining a graded module just as a family of modules $(M_{\la})$ instead of a direct sum. This is more or less the same and amounts to considering only homogeneous elements instead of  sums. The standard notions concerning direct sums, homomorphisms and tensor products in the graded situation are as follows. The   {direct sum}  of two graded modules (defined as usual) is naturally graded by $(M\oplus N)_{\la}=M_{\la}\oplus N_{\la}$.  A \emph{homomorphism} of graded modules $f\co M\to N$ \emph{of degree} $\mu$ is a collection of homomorphisms $f\co M_{\la}\to N_{\la+\mu}$ for all $\la$. (Sometimes notation $f=(f_{\la})$, where $f_{\la}\co M_{\la}\to N_{\la+\mu}$, is used.) Hence there are the family of modules $\Hom_{\mu}(M,N)$ of  homomorphisms of degree $\mu$ and the graded module $\Hom(M,N)=\oplus_{\mu} \Hom_{\mu}(M,N)$. We shall often refer to homomorphisms   simply as ``linear maps''. One may note that homomorphisms between graded modules can be seen as having naturally a bi-graded structure, i.e. $\Hom(M_{\la}, N_{\nu})$, which is then converted into a single ``total'' grading by considering $\Hom_{\mu}(M,N):=\prod_{-\la+\nu=\mu}\Hom(M_{\la}, N_{\nu})$. Likewise, the \emph{tensor product} $M\otimes N$ appears first as a bi-graded object, $M_{\la}\otimes N_{\nu}$, and then the {total degree}   is defined as the sum of the two degrees, so that $(M\otimes N)_{\mu}:= \oplus_{\la+\nu=\mu}M_{\la}\otimes N_{\nu}$. In the same way grading is defined for tensor products with any finite number of factors.

If there is a bilinear multiplication of any kind, it is naturally translated as a homomorphism $M\otimes N\to L$ (of some degree). Such are, in particular, multiplications in graded associative algebras and a left module structure over such an algebra. By default such  product structures are assumed of degree zero, i.e. $A_{\la}A_{\mu}\subset A_{\la+\mu}$ and $A_{\la}M_{\mu}\subset M_{\la+\mu}$.

The \emph{tensor algebra} of a graded module is naturally bi-graded by $(\ZZ,G)$: $T(M)=\oplus_{n=0}^{+\infty} T^n(M)$, where $T^n(M)=M^{\otimes n}= \oplus_{\la}(M^{\otimes n})_{\la}=
\oplus_{\la_1+\ldots+\la_n=\la}M_{\la_1}\otimes \ldots \otimes M_{\la_n}$. Definitions of the symmetric and exterior algebras depend on an extra piece of structure and will be discussed below.

\subsubsection{Dimension}

\label{subsubsec.dimens}

Suppose a graded module $M$ is \emph{free}, i.e. there is a basis consisting of free homogeneous generators. More precisely,   a basis is $E=\cup_{\la}E_{\la}$, where $E_{\la}\subset M_{\la}$, so that the  basis elements can be written as $e_{\la i_{\la}}\in E_{\la}$ where   $i_{\la}\in I_{\la}$ for some chosen set of indices $I_{\la}$, $|I_{\la}|=|E_{\la}|$. A free module is \emph{locally finite-dimensional}  if $|E_{\la}|<\infty$ for all $\la\in G$ and finite-dimensional if it is locally finite-dimensional and \emph{bounded} (i.e. $M_{\la}\neq 0$ only for a finite number of $\la$). (The forgetful functor maps free and finite-dimensional modules to the same type ungraded objects.) For a locally finite-dimensional module $M$, there are numbers $m_{\la}=\dim M_{\la}\in \ZP$, so there is a function $G\to \ZZ$, $\la\mapsto m_{\la}$, which we shall denote $\dimt(M)$  for the reasons which will be clear soon, $\dimt(M)(\la)=m_{\la}$. Then
\begin{equation}
    \dimt(M\oplus N)=\dimt(M)+\dimt(N)
\end{equation}
and
\begin{equation}
    \dimt(M\otimes N)=\dimt(M)*\dimt(N)\,,
\end{equation}
the convolution of functions, i.e.
\begin{equation}
    \dimt(M\otimes N)(\la)=\sum_{\mu}\dimt(M)(\mu)\dimt(N)(\la-\mu)\,.
\end{equation}
(The formula for the tensor product makes sense if one of the modules is finite-dimensional.) It is convenient to introduce \emph{formal exponentials} $e(\la)$ as symbols satisfying $e(\la+\mu)=e(\la)e(\mu)$ and define the (formal) \emph{dimension} of a locally finite-dimensional module as the formal sum
\begin{equation}
    \Dim(M)=\sum_{\la\in G} m_{\la} e(\la)\,.
\end{equation}
It is an element of the formal group ring $\ZZ[[G]]$ and the function
$\dimt(M)$ is the ``Fourier transform'' of $\dim(M)$. We have consequently
\begin{equation}
    \Dim (M\oplus N)=\Dim (M)+\Dim (N)
\end{equation}
and
\begin{equation}
    \Dim (M\otimes N)=\Dim (M) \Dim (N)\,,
\end{equation}
where the product is in the formal group ring. 

(Dimension of a graded module is very close to the notion of a formal character used in representation  theory, where they consider  formal exponents, see e.g. Kac~\cite{kac:bookinf}.)

\subsubsection{Grading and commutativity}

Grading in mathematics plays two different roles. One is a counting device, e.g. for ``vanishing by dimensional considerations''  type arguments  and   arguments by induction. Another is about the form of ``commutation rules''.
This concerns     linearity if the ground ring has a non-trivial grading,   relation between left and right modules, and   identification of $M\otimes N$ with $N\otimes M$. In all these cases there are elements   moving past each other and one needs rules  as to   what happens,  referred to as  a ``commutativity constraint''. (Abstract algebraic theory   is ``braidings in tensor categories'', but we do not have to go that far.) As it was discovered experimentally in topology, differential geometry and algebra, in the graded situation there   arises a non-trivial  ``commutativity constraint'' expressed by the following ``sign rule''. Fix a group homomorphism $\pr\co G\to \Z$, called \emph{parity}. For homogeneous elements $x$ of degree $\la$ we write $\pr(x)=\pr(\la)$. Then the \textbf{sign rule} says that if in the ungraded setting in a formula there is a swap of adjacent elements $x$ and $y$, then in the graded case it should be modified by inserting the sign $(-1)^{\pr(x)\pr(y)}$ (and any change of order is reduced to   swapping of neighbors). The observation is that if this rule is applied in the definitions, it will  then appear in the theorems. As Manin rightly notes, this cannot be made a meta-theorem because it fails for the graded analog of determinant; but it works in many cases. We quote some examples, which should be undoubtedly familiar, but we need them for reference.

Note that there is no change in the notions of a left or right module structures for the graded case (i.e. $a(bx)=(ab)x$ or $(xa)b=x(ab)$ is required) since there is no change in associativity laws.
If there are left module structures over a graded algebra $A$, a homomorphism of graded modules $f\co M\to N$ of degree $\la$ is (graded) \emph{$A$-linear} or \emph{$A$-linear from the left} or is an \emph{$A$-homomorphism} (a homomorphism over $A$)  if
\begin{equation}
    f(ax)=(-1)^{\pr(a)\pr(\la)} af(x)
\end{equation}
for a homogeneous $a\in A$. (In the sequel we shall always write formulas for homogeneous objects.) If there is a left module structure over $A$, then the \emph{multiplication from the right} is defined by
\begin{equation}
    xa:=(-1)^{\pr(x)\pr(a)} ax\,.
\end{equation}
One can immediately see that this formula gives a right module structure if $A$ is considered with the new product
\begin{equation}
    a\bullet b:= (-1)^{\pr(a)\pr(b)}ba\,.
\end{equation}
This is called  the \emph{opposite  algebra} $A^{\text{op}}$\,. Further, if a homomorphism $f$ is $A$-linear for left $A$-modules as defined above, this converts into
\begin{equation}
    f(xa)=f(x)a
\end{equation}
(no sign!), which should be taken as the definition of (graded) \emph{linearity from the right}. The \emph{commutator} of two elements in an algebra is defined by
\begin{equation}
    [a,b]:= ab - (-1)^{\pr(a)\pr(b)}ba\,.
\end{equation}
If if vanishes, the elements \emph{commute}. If all elements in an associate algebra commute, it is called \emph{commutative}. (We suppress any adjectives and do not say ``graded commutative''.) If an algebra is commutative, it coincides with the opposite algebra defined as above and hence left modules over a commutative algebra are also right modules, and conversely.  Properties of the commutator in an associative algebra  turned into axioms define what is traditionally called a ``graded Lie algebra''. The name may raise  some objections and we shall come back to that later (when we also elaborate the definition). Finally, a \emph{derivation}  of a graded algebra $A$ is defined as a homomorphism (linear map) $D\co A\to A$ of degree $\la$ (called the degree of a derivation) satisfying
\begin{equation}
    D(ab)=D(a)b +(-1)^{\pr(\la)\pr(a)} aD(b)\,.
\end{equation}
Similarly are defined derivations for other possible settings (e.g. over an algebra homomorphism $A\to B$). The guiding principle in all cases is that derivations are perturbations of maps respecting multiplication  (e.g algebra homomorphisms) and the rest follows from the formula for linearity.

The \emph{exterior algebra} $\L(M)$ and the \emph{symmetric algebra} $S(M)$ are defined as the quotients of the tensor algebra $T(M)$ by the ideals  generated by  $x\otimes y + (-1)^{\pr(x)\pr(y)} y\otimes x$ and $x\otimes y - (-1)^{\pr(x)\pr(y)} y\otimes x$ respectively (where $x,y\in M$). Recall that $T(M)$ is bi-graded by tensor degree and degree induced from $M$. Since these ideals are homogeneous in the ``bi-'' sense, both $\L(M)$ and $S(M)$ inherit bi-grading:   $\L(M)=\sum_n \L^n(M)=\sum_{n,\la} \L^n(M)_{\la}$ and $S(M)=\sum_n S^n(M)=\sum_{n,\la} S^n(M)_{\la}$. The symmetric algebra $S(M)$ is commutative with respect to the induced grading,
\begin{equation}
    a\cdot b = (-1)^{\pr(a)\pr(b)}\,b\cdot a\,,
\end{equation}
(tensor degree playing no role), while the exterior algebra  $\L(M)$ is not commutative. The multiplication in $\L(M)$ (called the exterior product) satisfies
\begin{equation}
    a\wedge b = (-1)^{\pr(a)\pr(b)+pq}\,b\wedge  a\,,
\end{equation}
where $a\in \L^p(M)$ and $b\in \L^q(M)$. Such a property is called \emph{skew-commutativity}.

\subsubsection{\mathversion{bold}Choice of grading group: $G=\ZZ\times \Z$}

So far we have worked with a general group $G$ endowed with a parity $\pr$. Favorite choices are: $G=\ZZ$ with $\pr(n)=n\! \mod 2$ and $G=\Z$ where $\pr$ is the identity. The former was a classical choice   in algebraic topology. The latter is the choice in superalgebra and supergeometry. We will use the following   combination: $G=\ZZ\times \Z$ and $\pr\co \ZZ\times \Z \to \Z$  is the projection on the second factor. (This includes other mentioned choices as special cases.) For the first factor $\ZZ$ we use the term \emph{weight}.  Notation: $\w(x)$ and $\pr(x)$ for the weight and parity of a homogeneous object $x$. We shall also use the tilde notation for parity, $\tilde x:=\pr(x)$. Objects of parity $0$ are referred to as \emph{even} and of parity $1$, \emph{odd}. We stress that  $\ZZ$ and $\Z$ gradings are in general independent; this does not exclude particular cases where parity of a given object  coincides with its weight modulo $2$. It may also happen that weights take values in $\ZZ^{+}$ only (i.e., there are no objects with negative weights). We have chosen the  term `weight' as generic; in particular situations the  $\ZZ$-grading may be called    `degree', `ghost number', etc.



Specifying the notion of graded dimension from~\ref{subsubsec.dimens} for the group $G=\ZZ\times \Z$, we arrive at \emph{graded dimensions} of the form
\begin{equation}
    \Dim M= \sum_{w\in \ZZ, \e=0,1} n_{w,\e}q^w \Pi^{\e} \in \ZZ[q,q^{-1}](\Pi)
\end{equation}
where $\Pi^2=1$. Recall that dimensions of $\Z$-graded objects take values in $\ZZ[\Pi]/(\Pi^2-1)=\ZZ+\ZZ\Pi$. Denote for convenience  the latter ring by  $\SZ$. Its elements   are written as $p+q\Pi$ or $p|q$, where $\Pi=0|1$ and $1=1|0$.  So we can re-write
\begin{equation}
    \Dim M= \sum_{w\in \ZZ} \hat n_{w}q^w \in \SZ[q,q^{-1}]\,.
\end{equation}

\subsubsection{Weight as   generalization of linear structure}

If we start from a vector space $V$ and consider an algebra generated by $V$ such as $T(V)$, $S(V)$ or $\L(V)$, they all have a natural $\ZZ$-grading by powers of $V$: $n$-fold products have degree $n$, etc. Linear maps of vector spaces induce algebra homomorphisms preserving degrees. Conversely, any algebra homomorphism of one of these algebras preserving degrees is naturally induced by a linear map of vector spaces. The key moment here is that all free generators of an algebra have the same (non-zero) grading, hence they cannot ``mix'' under grading-preserving homomorphisms. At the same time, allowing algebra homomorphisms of $T(V)$, $S(V)$ or $\L(V)$ without preservation of natural grading effectively destroys a memory of the original vector space $V$. Suppose however that free generators of one of these algebras are assigned weights not necessarily equal to each other. Then a homomorphism preserving such a weight is in general non-linear and a geometric object associated with such an algebra is no longer a vector space, but some ``graded manifold''.

Such a situation materializes, for example, for multiple vector bundles, say, double vector bundles  such as $TE$ or $T^*E$ for a given vector bundle $E\to M$. They are not vector bundles over the original base $M$ (no linear structure) and considering `total weight' for them leads to coordinates having weights $0, 1$, and $2$.

\subsection{Graded manifolds. Definition and constructions}

\subsubsection{Local models}

In the sequel, `polynomials' will refer to a free commutative algebra, which is an ordinary polynomial algebra if all generators are even,  a Grassmann algebra if all generators are odd, and the tensor product of an ordinary polynomial algebra and a Grassmann algebra in general.

To describe a local model of a graded manifold, consider a finite number of variables $x^a$ (practical needs my require considering the infinite-dimensional case as well, but here we confine ourselves to finite dimensions)  to which are assigned both parities and weights. Notation: $\w^a=\w(x^a)$, $\at=\pr(x^a)$. The variables $x^a$ are assumed to be commuting: $x^ax^b=(-1)^{\at\bt} x^bx^a$. We also assume that they are real, since we are going to define real graded manifolds. (Necessary modifications can be made for the complex or mixed complex-real cases.) The following classes of functions of   variables $x^a$ are natural to consider:
\begin{enumerate}[i)]
  \item polynomial  in all variables;
  \item smooth ($C^{\infty}$)  in the variables of weight $0$ and polynomial in the variables of weights $\neq 0$;
  \item smooth  ($C^{\infty}$)  in the variables of weight $0$ and formal power series in the variables of weights $\neq 0$.
\end{enumerate}
(Since the expansion  in a finite number of odd variables always terminates due to their nilpotence, and polynomials and smooth functions in odd variables are the same, the difference   arises only for even variables.)

Why one may need formal power expansions, is clear from the following example.

\begin{ex} Let $\w(x)=0$, $\w(y)=-1$ and $\w(z)=+1$. The element $x+yz$ is of weight $0$ and the substitution of it into any smooth function, e.g. $\sin x$, must be legitimate:
\begin{equation}
    \sin(x+yz)= \sin x + (yz)\cos x -\frac{1}{2}\,(yz)^2\sin x  +\cdots
\end{equation}
It transforms a smooth function of a variable of weight $0$ into a power series with respect to the variables $y, z$. (This should be contrasted with the case of odd variables $\x$ and $\h$, e.g.  $\sin (x+\x\h)=\sin x +  (\x\h) \cos x$, where the Taylor expansion always terminates.)
\end{ex}

Therefore, we are forced to consider formal power series in the variables of non-zero weights if we wish to use arbitrary smooth functions, not only polynomials, in the variables of weight zero, and if general transformations of variables preserving weights and parities are allowed. Hence the  two exceptions.

\begin{ex}[Restriction of admissible transformations] If we agree to restrict admissible transformations of variables    so that: (1) the variables of weight $0$ are allowed to transform only between themselves (no admixture of variables of weights $\neq 0$) and (2) within the variables of non-zero weights any homogeneous polynomial transformations are allowed, then the class of functions of $x^a$ that are arbitrary smooth in $x^a$ of $\w^a=0$ and polynomial in $x^a$ of $\w^a\neq 0$ will be stable under the corresponding substitutions. (Geometrically this corresponds to  considering fiber bundles where the variables of weight zero are coordinates on the base and the variables of non-zero weights are coordinates in the fibers.)
\end{ex}

\begin{ex}[Non-negative weights only]
\label{ex.nonnegweights}
Suppose for all variables $x^a$, $\w(x^a)\geq 0$. Then there is only a finite number of homogeneous monomials in $x^a$, $\w^a\neq 0$, of any given total weight $n\in \NN$ and there are no monomials of weight $0$. Hence all homogeneous formal power series in the  variables of non-zero weights are polynomials, and there can be no transformations admixing variables of non-zero weights to the variables of weight zero. (Geometrically  this corresponds to a fiber bundle  structure that comes about automatically.)
\end{ex}

Every algebra homomorphism from functions of $x^a$ to real numbers sends the variables that are of non-zero weight or odd, to zero. (They  cannot `take values' apart from zero.) On the other hand, all even $x^a$ such that $\w(x^a)=0$ can be considered as coordinates in the ordinary sense on some open domain (open set) $U_{\0}\subset \R{n_{\0}}$ that we can choose ($n_{\0}$ is the number of $x^a$ with $\w(x^a)=\pr(x^a)=0$).
We interpret $U_{\0}$ taken together with all our variables $x^a$, of all weights and parities, as a \emph{graded coordinate domain}, and use the same letter $U$ but without the subscript to denote this new object. The (graded) dimension of a graded coordinate domain $U$ is
\begin{equation}
    \dim U=\sum_w \hat n_wq^w \in \SZ[q,q^{-1}]\,,
\end{equation}
where $\hat n_w=n_{w,0}|n_{w,1}$ and $n_{w,0}$ (resp., $n_{w,1}$) is the number of even (resp., odd) coordinates $x^a$ of weight $w\in \ZZ$.
Formally, a graded coordinate domain $U$ can be seen as a pair consisting of an open domain $U_{\0}$ and the algebra $\Fune(U)$,
\begin{equation}
    \Fune(U):=\Fun(U_{\0})[[x^a\,|\, \w^a\neq 0 \  \text{or} \  \at=1]]
\end{equation}
(formal power series in variables of non-zero weights). By a slight modification, we can replace a single algebra $\Fune(U)$ by  a sheaf of algebras $\Funes_U$ on $U_{\0}$ (by taking $\Funes_U(V_{\0}):=\Fune(V)$ for all open subsets $V_{\0}\subset U_{\0}$) and define, finally,
\begin{equation}
    U:=(U_{\0},\Funes_U)\,.
\end{equation}
The sheaf $\Funes_U$ is  a sheaf of $\ZZ\times \Z$-graded commutative algebras over $\RR$ with a unit and there is a natural augmentation
\begin{equation}
    \e\co \Funes_U\to \mathscr{C}^{\infty}_{U_{\0}}
\end{equation}
given by sending all odd variables and all even variables of non-zero weight to $0$. On the stalk  $\Funes_{x_{\0}}$ at each $x_{\0}\in U_{\0}$ it gives a homomorphism  $\e\co \Funes_{x_{\0}}\to \RR$ and the kernel  $\Ker\e\subset \Funes_{x_{\0}}$ is a unique maximal ideal. We shall consider algebra homomorphisms $\Fune(W)\to \Fune(U)$ for graded coordinate domains that can be expressed by substitutions in coordinates, $y^i=\f^i(x)$, with the right-hand sides being functions of the same class (i.e. smooth in the even variables of zero weight and formal power series in the rest). (Possibly such are all the algebra homomorphisms for these algebras, but we do not want to dwell on that.) By augmentation, they induce algebra homomorphisms $\Fun(W_{\0})\to \Fun(U_{\0})$ and hence the usual smooth maps of the underlying coordinate domains
$\f_{\0}\co U_{\0}\to W_{\0}$. We define a \emph{morphism} (also called a \emph{smooth map}) between graded coordinate domains $\f\co U\to W$ as a morphism of local ringed spaces over $\RR$,
\begin{equation}
    \f=(\f_{\0},\f^*)\co (U_{\0},\Funes_U)\to (W_{\0},\Funes_W)\,,
\end{equation}
with the algebra homomorphisms
\begin{equation}
    \f^*\co \Funes_W(V_{\0}) \to \Funes_U(\f^{-1}(V_{\0}))
\end{equation}
for all open subsets $V_{\0}\subset W_{\0}$ being of described type.
Such morphisms are in a one-one correspondence with homomorphisms of algebras of ``global'' functions
\begin{equation}
    \f^*\co \Fune(W)\to \Fune(U)
\end{equation}
(The role of locality in the definition of a morphism is to ensure that the map $\f_{\0}$ of the underlying topological spaces is exactly the one obtained from the homomorphism of algebras with the help of   the augmentation.)

\subsubsection{Definition of a graded manifold}

The definition of  graded manifold is now completely straightforward; it mimics definitions of ordinary smooth manifolds and supermanifolds.  Recall that for a ringed space $X=(X_{\0},\mathscr{A}_X)$, an \emph{open subset} $U\subset X$  is the ringed space $U=(U_{\0}, \mathscr{A}_X|U_{\0})$, for an open  $U_{\0}\subset X_{\0}$. In the same sense we understand preimages and intersections of open subsets. An \emph{open cover} consists of open subsets $(U_{\a})$ such that $\bigcup U_{\a\0}=X_{\0}$. Fix a collection of numbers of variables with prescribed parities and weights, i.e., a  graded dimension. Consider graded coordinate domains of this graded dimension. Suppose $X$ is a local ringed space over $\RR$, so in particular $\mathscr{A}_X$ is a sheaf of $\RR$-algebras.  A local \emph{chart} for $X$ is an isomorphism  $\f\co V\to U$ of local ringed spaces over $\RR$, where $U\subset X$ is n open subset and $V$ is a graded coordinate domain. An \emph{atlas} for $X$ is a collection of charts $\f_{\a}\co V_{\a}\to U_{\a}$ such that $(U_{\a})$ make an open cover. We require that the resulting transformations of coordinates $\f_{\a\b}:=\f_{\a}^{-1}\circ \f_{\b}\co \f_{\b}^{-1}(U_{\a}\cap U_{\b})\to \f_{\a}^{-1}(U_{\a}\cap U_{\b})$ are smooth maps of graded coordinate domains. We refer to such atlases as smooth.  Two smooth atlases for $X$ are equivalent if their union is a smooth atlas.

\begin{de}
A smooth \emph{graded manifold}  of a given graded dimension  is a local ringed space over $\RR$, $X=(X_{\0},\mathscr{A}_X)$, with a Hausdorff second-countable underlying topological space $X_{\0}$, endowed with an equivalence class of smooth atlases.  The structure sheaf $\mathscr{A}_X$ will  be denoted by $\Funes_X$ or $\mathscr{C}^{\infty}_X$ and is called the sheaf of \emph{smooth functions}.  A \emph{smooth map} $f\co X\to Y$ of smooth graded manifolds is a morphism in the category of local ringed spaces over $\RR$ represented in local charts by smooth maps of graded coordinate domains.
\end{de}

Speaking informally, a graded manifold is a supermanifold with a distinguished class of atlases where coordinates   are additionally assigned weights in $\ZZ$ and the transformations of coordinates preserve both weights and parities. Smooth maps between graded manifolds are expressed in coordinates in the same way as for ordinary manifolds and supermanifolds. They are  formal power series in coordinates of non-zero weight.

Similarly  defined are  categories of graded manifolds in the mixed (real-complex) smooth and complex-analy\-tic settings.

\subsubsection{Simple examples. ``Graded sphere''.  ``Graded groups''}

Many examples of graded manifolds in applications arise from auxiliary constructions for ordinary manifolds. However, they can also arise in their own right as graded analogs of familiar  differential-geometric objects. The following example is meant to illustrate this point.
\begin{ex} Consider $\R{n(q)+1}$, where $n(q)=q^{-1}+n+q$, with coordinates  $x^1,\ldots,x^n,x^{n+1}$, $y$, $z$, of weights $\w(x^a)=0$, $\w(y)=-1$ and $\w(z)=+1$. Consider the equation
\begin{equation}
    (x^1)^2+\cdots+(x^n)^2+(x^{n+1})^2  +2yz=1
\end{equation}
(the left-hand side is a quadratic form of weight $0$). It specifies a \emph{graded sphere} $S^{n(q)}$ as a closed subspace $S^{n(q)}\subset \R{n(q)+1}$. Acting as for the ordinary sphere, one can introduce two charts $\f_N\co \R{n(q)}\to S^{n(q)}\setminus N$ and $\f_S\co \R{n(q)}\to S^{n(q)}\setminus S$, so that
\begin{equation}
\begin{aligned}
   &\f_N\co \  \xp=\frac{2\up_N}{|\up_N|^2+1}\,, \  x^{n+1}= \frac{|\up_N|^2-1}{|\up_N|^2+1}\,, \\
   &\f_N^{-1}\co \  \up_N=\frac{\xp}{1-x^{n+1}}\,,
\end{aligned}
\end{equation}
where $\xp=(x^a,y,z)$, $\up_N=(u^a_N, p_N, q_N)$, where $a=1,\ldots,n$, $\w(u^a_N)=0$, $\w(p_N)=-1$, $\w(q_N)=+1$, and $|\up_N|^2=$\linebreak $\sum_a (u^a_N)^2+2p_Nq_N$, and similar formulas for $\f_S$ (with the opposite sign for $x^{n+1}$). This gives
\begin{equation}
    \up_S=\frac{\up_N}{|\up_N|^2}
\end{equation}
(exactly as for the ordinary sphere or the supersphere) as the change of coordinates. This shows that $S^{n(q)}$ is a smooth graded manifold of dimension $n(q)=q^{-1}+n+q$. The underlying topological space of $S^{n(q)}$  
is the ordinary sphere $S^n$ of dimension $n$.  The algebra of smooth functions on $S^{n(q)}$ can be described  as the  ``inverse limit'':
\begin{equation}
\begin{aligned}
    &\fun(S^{n(q)})= \left\{f=(f_N,f_S)\in \fun(\R{n(q)}) \times \fun(\R{n(q)}) \,\Bigl|\right.\,\\
   &\kern3cm\left. \Bigr. f_N(\up_N)=f_S\bigl(\frac{\up_N}{|\up_N|^2}\bigr)\right\}\,.
\end{aligned}
\end{equation}
Note that the transformations of coordinates,\linebreak $u^a_S=u^a_N|\up_N|^{-2}$, $p_S=p_N|\up_N|^{-2}$, $q_S=q_N|\up_N|^{-2}$, where\linebreak $|\up_N|^{-2}=\bigl(\sum_a(u^a_N)^2\bigr)^{-1}\Bigl(1- \dfrac{2p_Nq_N}{\sum_a(u^a_N)^2} \,+ \ \cdots\Bigr)$,   are  formal power series in the coordinates of non-zero weights $p_N,q_N$.
\end{ex}

One can construct more similar examples as ``graded analogs'' of classical (super)manifolds. It would be interesting to study them systematically together with graded analogs of classical differential-geometric structures.


Another collection of examples can be obtained from graded Lie algebras (not to be confused with Lie superalgebras!). It is well known that $\ZZ$-gradings play important role in the  theory of finite- and infinite-dimensional Lie algebras.  Such algebras can come with natural gradings, which  are forgotten when the corresponding Lie groups are constructed. By taking these gradings into account one can obtain  `graded versions' of these groups.

\begin{ex} The vector space $\Mat(n)$ of (real) square $n\times n$ matrices becomes  $\ZZ$-graded if we take  the matrix units $E_{i,i+r}$, for a given $r\in\ZZ$, is a basis of  the subspace $\Mat(n)_r$ (which consists of matrices with non-zero entries only on the $r$th diagonal). Non-trivial graded components exist only for $|r|\leq n-1$. For each $r$, there are exactly $n-|r|$ elements on the $r$th diagonal, hence the graded dimension
\begin{equation}
    \dim \Mat(n) = \sum_{r=-n+1}^{n-1} (n-|r|) q^r = n+ \sum_{r=1}^{n-1} (n-r)(q^r+q^{-r})\,.
\end{equation}
This is a graded Lie algebra with respect to the matrix commutator:
\begin{equation}
    [\Mat(n)_r, \Mat(n)_s]\subset \Mat(n)_{r+s}\,.
\end{equation}
By multiplying the generators $e_{i_r}^{(r)}:=E_{i,i+r}$ of weight $r$ by parameters $t^{i_r}_{(r)}$ of weight  $-r$ (so to obtain an expression of weight zero) and taking the exponential,  we obtain an invertible matrix that can be regarded as a ``point'' of the  \emph{graded group}  (a group object in the category of graded manifolds) $\GL(n)_{\text{grad}}$ corresponding to the graded Lie algebra $\Mat(n)$,
\begin{equation}
    g=\exp \sum t^{i_r}_{(r)}e_{i_r}^{(r)}\,.
\end{equation}
Parameters $t^{i_r}_{(r)}$ are global coordinates on this graded manifold. Obviously, the underlying ordinary manifold is just the group of diagonal matrices 
with positive entries and the graded group $\GL(n)_{\text{grad}}$ can be regarded as its formal neighborhood in the  Lie group $\GL(n)$.
\end{ex}

\subsubsection{Constructions with graded manifolds}
\label{subsec.constrgrad}

There are obvious analogs of constructions for ordinary manifolds and supermanifolds, such as submanifolds, products, etc. Closed  submanifolds are locally specified by systems of equations of constant rank. It is required that the equations be homogeneous both in parity and weight, where the notion of rank is understood as `graded rank'. Then the dimension of $S\subset X$ is $\dim S=\dim X -r$, where $r=r(q)\in \SZ[q,q^{-1}]$ is the   rank of the system of equations.  As for the product $X\times Y$ of graded manifolds, it is most natural to consider it as bi-graded. (``Bi-grading'' refers to two $\ZZ$ gradings, with a single parity.))

Every vector bundle by default can be considered as a graded manifold so that linear coordinates in the fibers are assigned weight $+1$. Then fiberwise linear maps are the same as weight-preserving.

All objects on a graded manifold assume weights, e.g. tangent vectors, covectors, vector fields, etc.
Tangent and cotangent bundles for graded manifolds   carry a bi-grading. One $\ZZ$-grading (actually, $\ZZ^{\geq 0}$) is the vector bundle grading by degree in   fiber coordinates. Another is the induced weight.

\begin{example}
If $x^a$ are coordinates of weights $w^a$, the partial derivatives $\lder{}{x^a}$ have weights $-w^a$. Hence we assign weights $-w^a$ to the momentum variables $p_a$ canonically conjugate to $x^a$. We arrive at the cotangent bundle $T^*M$ for  a graded manifold $M$ as a bi-graded manifold. The first grading induced from $M$ we continue to call \emph{weight} and it is given by $\w(x^a)=w^a, \w(p_a)=-w^a$. The second grading we call \emph{degree}  and it just expresses the vector bundle structure: $\deg x^a=0, \deg p_a=+1$.
\end{example}

\begin{example} For a vector bundle $E\to M$ regarded as a graded manifold in the usual way, its cotangent bundle $T^*E$ is a double vector bundle~\cite{mackenzie:book2005}, with the side bundles $E\to M$ and $E^*\to M$:
\begin{equation}
    \begin{CD} T^*E  @>>>   E^*\\
                @VVV         @VVV\\
                  E @>>>  M
    \end{CD}
\end{equation}
(this is related with the Mackenzie--Xu theorem,   see~\ref{subsubsec.adjoint}). The double vector bundle structure gives two gradings on $T^* E$\,: $\w_1=\#p_a+\#p_i$ and $\w_2=\#p_a+\#u^i$. Here we denote by $x^a$ coordinates on the base, by $u^i$ coordinates in the fibers of $E$, and by $p_a$, $p_i$ the conjugate momenta. Compared to out previous analysis, we have $\w=\#u^i-\#p_i=\w_2-\w_1$ as   induced weight and $\deg=\#p_a+\#p_i=\w_1$ as   degree.
\end{example}

(In physics, the above grading $\w=\#u^i-\#p_i$ appears under the name ``ghost number'', see~\cite{henneaux:teitelboimbook}.)

Another example is provided by differential forms on a vector bundle. Recall that in supergeometry, pseudodifferential forms, which we with an abuse of language will call simply ``forms'', are functions on the antitangent bundle.

\begin{example}
Consider for a vector bundle $E$, its antitangent (parity reversed  tangent) $\Pi TE$. It is again a double vector bundle
\begin{equation}
    \begin{CD} \Pi TE  @>>>   \Pi TM\\
                @VVV         @VVV\\
                  E @>>>  M
    \end{CD}
\end{equation}
It has two weights corresponding to the two vector bundle structures: $\w_1=\#{\rm d}x^a +\#{\rm d}u^i$ and $\w_2=\#u^i+\#{\rm d}u^i$. Here     induced weight is $\w=\#u^i+\#{\rm d}u^i=\w_2$ and degree is $\deg =\#{\rm d}x^a +\#{\rm d}u^i=\w_1$. In~\cite{tv:ivb,tv:class} we discovered and used grading $\w_1 -\w_2=\#{\rm d}x^a-\#u^i$ on forms on $E$. Note that the de Rham differential has degree $+1$ in this grading.
\end{example}

Besides grading, manifolds can be endowed with a filtration. For example, for a bi-graded manifold, one of the gradings can become a filtration if more general transformation are considered. Such is the case of ``resolution degree'' in~\cite{henneaux:teitelboimbook}, which is preserved only as a filtration under canonical transformations.  We do not formalize \emph{filtered manifolds} here, since this notion should be clear.

\subsubsection{Structure of a graded manifold}
Let $X=(X_{\0}, \Funes_X)$ be a graded manifold. Denote by $\mathscr{J}_X:=(\Funes_X)_{\neq 0}+(\Funes_X)_{\neq 0}^2$ the   ideal  generated by all functions of non-zero weight. Its zero locus is a closed submanifold $X_0=(X_{\0}, \Funes_{X_0})$, $\Funes_{X_0}=\Funes_X/\mathscr{J}_X$, with the same underlying topological space $X_{\0}$. Note that in general $X_0$ is a supermanifold.  It should not be confused with $X_{\0}$, which has a natural   structure of an ordinary manifold, $X_{\0}=(X_{\0}, \Funes_{X_{\0}})$, where $\Funes_{X_{\0}}=\Funes_X/((\Funes_X)_{\neq 0}+(\Funes_X)_{\text{odd}})$\,. In general, $X_{\0}\neq X_0$, only  $X_{\0}\subset X_0$.  Only if there are no odd coordinates of zero weight, then $X_0$ is an ordinary manifold and $X_0=X_{\0}$. We will be more concerned with $X_0$. Powers of the ideal $\mathscr{J}_X$ define infinitesimal neighborhoods $X_k=(X_{\0},   \Funes_X/\mathscr{J}_X^{k+1})$, of the closed   submanifold $X_0\subset X$. This is an infinite sequence and $X$ is its direct limit:
\begin{equation}
    X_0\subset X_1 \subset \ldots X_k\subset X_{k+1} \subset \ldots \ldots X\,.
\end{equation}

Consider the normal bundle to $X_0$ in $X$, defined as usual as the quotient $(TX|_{X_0})/TX_0$. Denote it $N$. To see its structure, denote local coordinates of weight zero on $X$ by $x^a$ and local coordinates of non-zero weights, by $y^i$. Transformation of coordinates has the form
\begin{equation}
\begin{aligned}
    x^a&=x^{a}(x', y')\,,\\
    y^i&=y^i(x',y')\,,
\end{aligned}
\end{equation}
where $x'$ and $y'$ denote `new' coordinates of weights $0$ and $\neq 0$, respectively, $x^{a'}$ and $y^{i'}$. The right-hand sides are formal power series in coordinates of non-zero weights. For the induced transformation of fiber coordinates in the tangent bundle we obtain
\begin{equation}
\begin{aligned}
    \dot x^a&=\dot x^{a'}\der{x^{a}}{x^{a'}}(x', y')+ \dot y^{i'}\der{x^{a}}{y^{i'}}(x', y')\,,\\
    \dot y^i&=\dot x^{a'}\der{y^{i}}{x^{a'}}(x', y')+ \dot y^{i'}\der{y^{i}}{y^{i'}}(x', y')\,.
\end{aligned}
\end{equation}
Note that the Jacobi matrix for transformation of coordinates on $X$ is a block matrix with blocks numbered  by weights. In the formulas above, the matrix $(\der{x^{a}}{x^{a'}})$ is the zero-zero block of the Jacobi matrix, the matrix $(\der{y^{i}}{y^{i'}})$ consists of possibly several other diagonal blocks, while the matrices $(\der{x^{a}}{y^{i'}})$ and $(\der{y^{i}}{x^{a'}})$ consist of off-diagonal blocks. The entries in the diagonal blocks have weight zero; the entries in the off-diagonal blocks are of corresponding non-zero weights. The diagonal blocks are invertible and, as formal power expansions in   variables of no-zero weights,   will remain invertible is all such variables are set to zero.  (One may say that the Jacobi matrix for transformation of coordinates on a graded manifold takes values in the \emph{graded general linear group}.) Upon restriction to $X_0$, all coordinates of non-zero weights $y^i$ become zero and, in particular, all elements of off-diagonal blocks of the Jacobi matrix will vanish. Hence the transformation of coordinates on the normal bundle $N$ as a vector bundle over $X_0$ will be
\begin{equation}
    x^a_0=x^{a}(x'_0, 0)
\end{equation}
for the coordinates on the base $X_0$, which we have marked with the subscript $0$ to distinguish them from those on $X$, and
\begin{equation}
    \dot y^i_0 = \dot y^{i'}_0\,\der{y^{i}}{y^{i'}}(x'_0, 0)
\end{equation}
for the fiber coordinates, where likewise we have attached the subscript to distinguish them from (a part of) coordinates on the tangent bundle $TX$. The normal bundle $N$ is a $\ZZ$-graded vector bundle over a non-graded base $X_0$, so it is a direct sum of ordinary vector bundles (with assigned weights). (Everything is in the category of supermanifolds, which makes no real difference here.) If we treat $N$ as a graded manifold itself, it has the same graded dimension as $X$.

Note now that the graded manifold $X$ is formal in the directions normal to the submanifold $X_0$. (Changes of coordinates are given by formal power series in coordinates of non-zero weights.) In the same way as for supermanifolds and their underlying ordinary manifolds,  there are no actual (non-infinitesimal) `intermediate' neighborhoods between $X_0$ and $X$. Therefore, in the smooth case, exactly as for smooth supermanifolds, an analog of the tubular neighborhood theorem gives a (non-canonical) \emph{diffeomorphism}
\begin{equation}
    X\cong N
\end{equation}
\emph{as graded manifolds}. This is the \textbf{classification theorem for smooth graded manifolds}.

In the complex-analytic case, one should expect an \textbf{analog of Vaintrob's theorem}~\cite{vaintrob:viniti} for complex-analytic supermanifolds: namely, that a complex-analytic graded manifold  $X$ is a deformation of the respective normal bundle $N$.

In the same way as for smooth supermanifolds, the possibility to describe a smooth graded manifold as a graded vector bundle over an ordinary base (non-canonically), does not make smooth graded manifold not interesting. The key difference between graded vector bundles and graded manifolds is that the latter have more morphisms as transformations mixing variables of different weights with the only condition that the total weight\,---\,as well as parity\,---\,be preserved . Note also that, even more, such morphisms can themselves depend on parameters of non-zero weight leading to \emph{graded manifolds of maps}, in particular already mentioned \emph{graded groups}, etc. etc.

A remark giving a different perspective is that, in some cases, the formal power series defining transformations of variables in a graded manifold $X$ can happen to be the Taylor series of genuine smooth transformations of even coordinates in an ordinary manifold $\tilde X$. (Such are the above examples of the ``graded sphere'' $S^{n(q)}$ and the graded group $\GL(n)_{\text{grad}}$.) In general, one can see a graded manifold as a formal germ of an ordinary (super)manifold of the ``ordinary'' dimension $n$ obtained as $n=n(1)$ for $\dim X=n(q)$.

\subsubsection{Graded manifolds of maps. Functor of points}

One may wish to consider ``graded  manifolds of mappings''. First of all, they have to be (in general) infinite-dimensional, hence strictly speaking   outside the scope of the definition above. Difficulty with introducing such objects comes from  two separate but entangled causes. One is their infinite-dimensionality, and fundamentally this is the same difficulty that we have for ordinary manifolds
when we want to define a manifold of maps.  
The other cause of the difficulty is of `graded' nature and
has to be overcome already  for supermanifolds.
It is resolved by allowing for odd as well as non-zero-weight parameters  in the formulas for the mappings (possibly, infinite number of them). In brief, the \emph{graded manifold of maps} $\Mapp(X,Y)$ for finite-dimensional graded manifolds $X$ and $Y$ is defined ``in the weak sense'' by the formula
\begin{equation}
    \Map(Z, \Mapp(X,Y))=\Map(Z\times X, Y)\,,
\end{equation}
where $\Map$ stands for the set of morphisms in the category of graded manifolds and $Z$ is an arbitrary graded manifold. The equality should be understood as an isomorphism of functors. The right-hand-side serves as the definition of the left-hand-side, i.e., $\Mapp(X,Y)$ is defined as the representing object for the functor $Z\mapsto \Map(Z\times X, Y)$ (if existed). In other words, the functor is known and we work with it as if it were representable (this is what is meant by ``weak sense''). The meaning of the above formula is that, for a given graded manifold $Z$, we consider all maps $X\to Y$ depending on coordinates on $Z$ as external parameters; then  $\Mapp(X,Y)$, if one can define it, serves as the ``universal family'' of maps and coordinates on it are ``universal'' parameters.

Acting naively, we can describe the graded manifold $\Mapp(X,Y)$ as follows. If $x^a$ and $y^i$ are, respectively, local coordinates on $X$ and $Y$, then ``coordinates'' on $\Mapp(X,Y)$ are functions $y^i=\f^i(x)$, $x=(x^a)$, defined by expansions over odd variables and variables of non-zero weights, where the coefficients of the expansions, which should be ordinary smooth functions of the coordinates $x^a$ of weight zero, are treated formally as having the required  parities and weights (possibly, non-zero).

In the following two examples   we can avoid, or partly avoid, the problem arising from infinite-dimensionality.

\begin{ex} For any graded manifold $X$,
\begin{equation}
    \Mapp(\R{0|1}, X)=\Pi TX\,.
\end{equation}
This is well known (at least in the non-graded case). Indeed, if $x^a$ are local coordinates on $X$ and $\tau$ is the single coordinate on $\R{0|1}$, $\pr(\tau)=1$, $\w(\tau)=0$, then ``coordinates'' on $\Mapp(\R{0|1}, X)$ are functions of $\tau$,
\begin{equation}
    x^a=\f(\tau)=\f^a_0 + \tau \f_1^a\,,
\end{equation}
where $\w(\f^a_0)=\w(\f^a_1)=\w(x^a)$, $\pr(\f^a_0)=\pr(x^a)$, and $\pr(\f^a_1)=\pr(x^a)+1$\,. By checking the transformation law, one can immediately identify the variables $\f^a_0$ and $\f_1^a$ with $x^a$ and $dx^a$, respectively, the latter considered as coordinates on $\Pi TM$.
\end{ex}

\begin{ex} Consider an even variable $t$ of weight $-1$ as a coordinate on $\R{q^{-1}}$. Find the graded manifold $\Mapp(\R{q^{-1}}, X)$, for an arbitrary graded manifold $X$. The ``coordinates'' on $\Mapp(\R{q^{-1}}, X)$ will be the power series
\begin{equation}
    x^a=\f^a(t)= \sum_{n=0}^{+\infty} \frac{1}{n!}\,t^n \f^a_n\,,
\end{equation}
with indeterminate coefficients $\f^a_n$, where $\pr(\f^a_n)=\pr(x^a)$  and $\w(\f^a_n)=\w(x^a)+n$, for all $n=0,1, 2, \ldots \ $. Although in this case the graded manifold of maps is infinite-dimensional (unlike the previous example), its  infinite-dimensionality is easily controllable. The transformation law for the variables $\f^a_n$ follows from the  expansion in $t$ of  $x^a=x^a(\f_0'+ t\f_1' + \ldots)$, where $x^a=x^a(x')$ is a change of coordinates on $X$. One gets
\begin{equation}
\begin{aligned}
    \f_0^a&= x^a(\f_0')\\
    \f^a_1 &= \f^{a'}_1\, \der{x^a}{x^{a'}}(\f'_0)\\
    \f^a_2&= \f^{a'}_1\f^{b'}_1\,\dder{x^a}{x^{b'}}{x^{a'}}(\f'_0) + \f^{a'}_2\, \der{x^a}{x^{a'}}(\f'_0)\\
    \dots
\end{aligned}
\end{equation}
The transformation law for the variables of weight $n$ involves only variables of weights $\leq n$, so can be truncated at any $n$. The graded manifold $\Mapp(\R{q^{-1}}, X)$ is the inverse limit of finite-dimensional graded manifolds. We recognize in  $\Mapp(\R{q^{-1}}, X)$ the infinite-order tangent bundle $T^{(\infty)}X$, which is the limit of higher tangent bundles $T^{(N)}X$, taken with their natural graded structures.  (Or spaces of   jets of parameterized curves in $X$.)
\end{ex}

What about functions on a graded manifold $X$? Can they be fit into the above?

The graded manifold $\Mapp(X,\RR)$ is the manifold of all even functions of weight zero. To odd functions or functions of non-zero weight, one needs to consider $\Mapp(X,\Pi \RR)$, $\Mapp(X,\RR[n])$ or $\Mapp(X,\Pi \RR[n])$. These are linear graded manifolds and we have $\Mapp(X,\Pi \RR)=\Pi \Mapp(X,\RR)$, etc.

The way how graded manifold of maps is introduced is an example of  the idea the ``functor of points''. Its origins are in algebraic geometry and it is well known for supermanifolds. Namely, every graded manifold $X$ defines a contravariant functor on the category of graded manifolds, $Z\mapsto \Map (Z,X)$ (to the category of sets). Elements of the set $\Map (Z,X)$ are called \emph{$Z$-points} of $X$. In particular, the points of the underlying topological space $X_{\0}$ are exactly  $\R{0}$-points of $X$. A graded manifold is completely defined by its functor of points. Hence the general idea when a graded manifold is being looked for in  some problem, to find it first ``in the weak sense'', i.e. introduce first a functor that should serve as it functor of points for   graded manifold in question and then see if it is indeed representable. (We borrowed the analogy with a ``weak solution'' of a  differential equation  from K.~Fukaya.)

\subsection{Non-negatively graded manifolds}

\subsubsection{Non-negatively graded manifold as a fiber bundle}
\label{subsubsec.nonegasbundle}
Suppose all local coordinates $x^a$ for a graded manifold $E$ have non-negative weights. Then: transformations of coordinates cannot have infinite power series and are necessarily polynomial (see Example~\ref{ex.nonnegweights}). We can arrange   coordinates by increasing weights. Then the coordinates of zero weight transform  between themselves, coordinates of weight $+1$ undergo linear transformation, coordinates of weight $+2$ transform linearly between themselves but can have a  term quadratic in coordinates of weight $+1$, etc. We arrive at a canonical tower of fibrations:
\begin{equation}\label{eq.tower}
    E=E_N\to E_{N-1}\to \cdots\to E_2\to E_1\to E_0=M
\end{equation}
Here $N$ is the top weight of local coordinates on $E$. The subscript for $E_k$ means the top weight for $E_k$. The first fibration $E_1\to E_0=M$ is a vector bundle, the rest are affine bundles. Altogether this assembles into a fiber bundle $E\to M$ with special form polynomial transition functions. This picture was introduced in~\cite{tv:graded}.

The standard fiber for $E\to M$ is some $\R{D}$ where $D=\sum_{w>0} n_wq^w$, i.e., an affine space with coordinates assigned with  some positive weights. Denote this model graded space by $V$. Denote by $GG(V)$ the group of graded polynomial transformations of $V$ (``general  graded''). It depends only on dimension of $V$.

Non-negatively graded manifolds because of restrictions posed by their bundle structure are particularly useful for encoding various differential-geometric information. The method is placing a bound on ``height'' (top weight of local coordinates) combined with ``component analysis'' of some graded quantity. The simplest but still very useful case is to treat vector bundles as graded manifolds. This helps e.g for description of Lie algebroids and multiple Lie algebroids (see below in \ref{subsubsec.lalg}). (Also the description of Courant algebroids by Roytenberg~\cite{roytenberg:thesis}, \cite{roytenberg:graded}.)

 \subsubsection{Canonical  linear model}

Non-negatively graded manifolds  are the most direct  generalization of    vector bundles.   There is one problem related with the fact that, unlike vector bundles, sections of such a nonlinear bundle cannot be added or multiplied by numbers, so we seem to lose an algebraic arena  where algebraic structures such as brackets can be defined.  We shall show here that for   a non-negatively graded manifold $E$ regarded as a fiber bundle $E\to M$, its  structure group  $GG(V)$
possesses  a  natural faithful finite-dimensional linear representation $\rho$. It plays the role of  the standard representation  of the general linear group, and reduces to it in the linear case.
The associated vector bundle $\rho(E)\to M$
can be seen as a canonical ``linearization'' of the graded manifold $E$. It can be defined directly as corresponding to the projective module $\Vectn(E)$  over $\fun(M)$ consisting of vector fields on $E$ of negative weight.

Consider the standard fiber $V$, which is positively graded, and the space of vector fields. It naturally expands by weights as $\Vect (V)=\Vectn(V)\oplus \Vectp(V)$,
\begin{equation}
    \Vectn(V)=\Vect_{-N}(V)\oplus\cdots \oplus \Vect_{-1}(V)\,.
\end{equation}
The group $GG(V)$ acts on $\Vectn(V)$.
\begin{thm} The representation of $GG(V)$   on $\Vectn(V)$ is faithful.
\end{thm}
\begin{proof} On the infinitesimal level, if a vector field of zero weight commutes with all vector fields of negative weights, then in particular it commutes with all partial derivatives $\lder{}{y^i}$, therefore it has constant coefficients, hence is zero.
\end{proof}

We call the representation of $GG(V)$   on $\Vectn(V)$ the \emph{fundamental representation}.

It is finite-dimensional. In the   case of $GL(n)$  it is  the standard    representation on $\R{n}$.

The associated   bundle $F\to M$ corresponding to the fundamental representation of the group $GG(V)$   is called the \emph{fundamental vector bundle} of the graded manifold $E$.

Its sections can be identified with $\Vectn(E)$. Since the representation $\rho$ is faithful, the bundle $E\to M$ (its transition functions) can be recovered from the vector bundle $F\to M$. We shall use this vector bundle when considering ``non-linear Lie algebroids'' in~\ref{subsubsec.nonlinlalg}.

\subsection{Historical remarks about  graded  notions}

Graded notions have long played important role in different areas of mathematics, from gradings appearing in the theory of Lie algebras where it was used as a tool in classification and e.g. for measuring growth of infinite-dimensional algebras, to graded objects in topology and differential geometry, where grading was used for induction and ``dimensional'' arguments and as a source of ``sign rule''.  Nijenhuis--Richardson~\cite{nijen:deformm} developed the basics of graded algebras using  grading by an arbitrary abelian group endowed with a parity homomorphism. (They also anticipated Lie supergroups.)

Looking at other important works, we may notice that probably until the physics works related with BRST quantization, $\ZZ$-grading for algebras was almost always  either $\ZZ^{\geq 0}$ or $\ZZ^{\leq 0}$. Tate~\cite{tate:homology1958} uses  non-negatively graded algebras with homological differential. Milnor and Moore~\cite{milnor:andmoore} by ``graded'' mean  $\ZZ^{\geq 0}$-graded. Deligne--Griffiths--Morgan--Sullivan~\cite{deligne-ea:realhomotopy1975}   and Sullivan~\cite{sullivan:infinites1977} saw graded algebras as non-negatively $\ZZ$-graded.
Quillen~\cite{quillen:hoalgebra67} uses non-negative grading.
Boardman  says quite explicitly that ``graded'' means $\ZZ^{\geq 0}$-graded~\cite{boardman:signs1966} and proceeds to establishing the sign rule as a precise theorem.

Berezin, working on implementation of his program of supermathematics (before the name) made a decisive step in   separating   $\Z$-grading responsible for signs from $\ZZ$-grading. In particular, he studied automorphisms of Grassmann algebra as a $\Z$-graded algebra~\cite{berezin:autgrass} and on these paths discovered Berezinian. Without that, there would be no supermanifolds. (Supermanifolds were\linebreak briefly known for some as ``graded manifolds'' following Kostant, but this usage has now gone.) Re-introduction of $\ZZ$-grading into supergeometry, in a different way, is a new turn of Hegel's dialectic  spiral.

Schlessinger and Stasheff in their famous long-secret work~\cite{stasheff:schless} use graded as $\ZZ$-graded while noting that in many cases it will be either $\geq 0$ (as cochains in topology) or $\leq 0$ (as in algebraic geometry); they however consider a bi-graded case for the ``Tate--J\'{o}zefiak resolution''.  (J\'{o}zefiak~\cite{jozefiak:tateresol1972} generalized Tate's resolution to graded case.)

(Working in a different area, the present author found a non-standard $\ZZ$-grading for pseudodifferential forms on a vector bundle and used it for a study of integral transforms~\cite{tv:ivb},~\cite{tv:class}.)

Supermanifolds graded additionally by $\ZZ$ (and sometimes endowed with several gradings and/or filtration) appeared without  any particular name or mathematical formalization in Henneaux--Teitelboim~\cite{henneaux:teitelboimbook}. They however were quite explicit that parity and $\ZZ$-grading  (e.g. ``ghost number'') are independent and the latter can be positive and negative.

Kontsevich in~\cite{Kontsevich:1997vb} introduces the tensor category of ``graded vector spaces'' as a full subcategory of $\ZZ$-graded super vector spaces for which parity equals degree mod $2$ and also ``graded manifolds'' as supermanifolds with extra  $\ZZ$-grading in the structure sheaf with the same restriction. (He commented also that many of his constructions are valid just for supermanifolds and do not require $\ZZ$-grading.) \v{S}evera~\cite{severa:some2005} introduces a version of $\ZZ^{\geq 0}$-graded manifolds where parity equals degree mod $2$ and they become popular especially combined with a $Q$-structure under the name $NQ$-manifolds ($N$ presumably for $\NN$). Graded manifolds as defined here (with  $\ZZ\times \Z$-grading) were introduced and studied in~\cite{tv:graded}. In particular, the tower of fibrations~\eqref{eq.tower} for non-negatively graded manifolds appeared there. We have used them as a standard language ever since, see e.g.~\cite{tv:qman-esi},~\cite{tv:higherpoisson},~\cite{tv:napl},~\cite{tv:qman}, ~\cite{tv:qman-mack}.


\section{Language of $Q$-manifolds. Description of algebraic and geometric structures}
\label{sec.qman}

In this section we will introduce the language of $Q$-manifolds, which are supermanifolds endowed with an odd vector field  of square zero. They provide a  powerful tool  for describing differential-geometric and algebraic structures. From the viewpoint of algebraic-geometric duality, $Q$-manifolds are the geometric counterpart of differential $\Z$-graded algebras and can be seen as a basis of a ``non-linear homological algebra''.

\subsection{Definition of a $Q$-manifold. Main notions}

\subsubsection{Definition and model examples}
\begin{definition}
A \emph{$Q$-manifold} is a supermanifold endowed with an odd vector field $Q$ such that $Q^2=0$. Such a vector field is called \emph{homological}. A homological vector field is also referred to as a \emph{$Q$-structure}.
\end{definition}

Note that for an odd $Q$, $Q^2=\frac{1}{2}[Q,Q]$. We may sometimes write a $Q$-manifold as a pair $(M, Q)$.

If $M$ is a $Q$-manifold and $x^a$ are local coordinates on $M$, so that $Q=Q^a(x)\lder{}{x^a}$, the condition $Q^2=0$ is expressed by
\begin{equation}\label{eq.qsquare}
    Q^a\p_aQ^b=0\,.
\end{equation}

\begin{remark}
The notion of a $Q$-manifold was introduced by A.~S.~Schwarz, see~\cite{schwarz:semiclassical}. The notation $Q$, can be traced  back to the earlier study of supersymmetry in physics, where the letter $Q$ was a standard notation for a supercharge, i.e. an odd operator such that $Q^2=H$, where $H$ is the (quantum) Hamiltonian or more generally an even symmetry generator. If such an even symmetry vanishes for whatever reason, we arrive at the situation when $Q^2=0$ (see e.g.~\cite{Witten:1988xj}). Homological vector fields were studied by Vaintrob~\cite{vaintrob:darboux, vaintrob:normforms}. Seminal role was played by the work of Alexandrov--Kontsevich--Schwarz--Zaboronsky (AKSZ)~\cite{schwarz:aksz} and the application of $Q$-mani\-folds by Kontsevich in~\cite{Kontsevich:1997vb}. But before $Q$-manifolds were formalized as a mathematical notion, homological vector fields had existed in physics as \emph{BRST symmetries} (for Becchi--Rouet--Stora and I.~Tyutin), see monograph~\cite{henneaux:teitelboimbook}. The physicists' approach for a long time was only half-geometrical, as they were mainly drawing from known algebraic methods of homological algebra (e.g. Tate resolution).
\end{remark}

\begin{example} For any (super)manifold $M$, the supermanifold $\Pi TM$ is a $Q$-manifold. The $Q$-structure is given by the de Rham differential:
\begin{equation}
    Q={\rm d}={\rm d}x^a\der{}{x^a}\,.
\end{equation}
This example from many viewpoints plays the same role for $Q$-manifolds as $T^*M$ with the canonical symplectic structure plays for symplectic manifolds. (Also, we shall see that from some abstract viewpoint, the $Q$-structure on $\Pi TM$ as well as the even and odd symplectic structures on $T^*M$ and $\Pi T^*M$ are manifestations of ``one and the same structure''.)
\end{example}

\begin{example}
Let $V$ be a $\Z$-graded vector space, which we treat as a supermanifold and actually as a graded manifold (in the usual way). An odd differential on $V$, i.e. an odd linear operator ${\rm d}\co V\to V$ such that ${\rm d}^2=0$ defines a ``linear vector field''
\begin{equation}
    Q=x^a{\rm d}_a{}^b \der{}{x^b}\,,
\end{equation}
(where $({\rm d}_a{}^b)$ is the matrix of the linear operator ${\rm d}$), which is a $Q$-structure. Note that $w(Q)=0$ for the natural $\ZZ$-grading. If $V$ is a cochain complex, i.e. is itself endowed with a $\ZZ$-grading so that $\deg {\rm d}=+1$, the corresponding supermanifold becomes bi-graded (by weight and by  degree), and $\deg Q=+1$.
\end{example}

\begin{example}
Let $\frakg$ be a Lie algebra (we will shortly generalize to Lie superalgebras). Consider the supermanifold $\Pi \frakg$. Let $\x^i$ be linear coordinates on $\Pi \frakg$ corresponding to a basis $e_i$ in $\frakg$. (Because $\frakg$ is purely even, all coordinates $\x^i$ are odd.) Consider a vector field on $\Pi \frakg$
\begin{equation}
    Q=\frac{1}{2}\x^i\x^jc_{ij}^k\der{}{\x^k}\,,
\end{equation}
where $c_{ij}^k$ are the structure constants of $\frakg$ in the basis $e_i$. The vector field $Q$ is odd and $w(Q)=+1$ w.r.t. grading given by the linear structure. One can check that $Q^2=0$ due to the Jacobi identity for $c_{ij}^k$. Moreover, the condition $Q^2=0$ is exactly equivalent to the Jacobi identity in $\frakg$.
\end{example}

\begin{remark} The previous example is classical. Functions on $\Pi \frakg$ can be identified with the ``standard cochain complex'' $C^*(\frakg)$ of a Lie algebra $\frakg$ and the vector field $Q$ is the Chevalley--Eilenberg differential in this complex. We shall use this example as a model for describing other structures. One should also compare the formula for the vector field $Q$ on $\Pi \frakg$ with the formulas
    $\{y_i,y_j\}=c_{ij}^k y_k$
for the Lie--Poisson bracket (=Berezin--Kirillov bracket) on $\frakg^*$ and
   $ \{\h_i,\h_j\}=c_{ij}^k \h_k$
for the Lie--Schouten bracket (odd analog of Lie--Poisson) on $\Pi\frakg^*$. For a Lie algebra, these three structures are different equivalent manifestations of a Lie algebra structure itself. If we drop the restrictions e.g. the linearity for the brackets, we will arrive at $Q$-manifolds, Poisson manifolds and odd Poisson manifolds as three different non-linear generalizations of Lie algebras.
\end{remark}

The general philosophy is that a $Q$-manifold  is a non-linear analog  of a (co)chain complex.  Respectively, we will introduce now the analogs for chain maps and for cohomology. We will also give the analog of the complex of homomorphisms.

\subsubsection{$Q$-morphisms}

\begin{definition} A \emph{morphism} of $Q$-manifolds $(M_1, Q_1)$ to $(M_2, Q_2)$ or a \emph{$Q$-morphism} or a \emph{$Q$-map} is a supermanifold map $\f\co M_1\to M_2$ that intertwines $Q_1$ and $Q_2$, i.e. such that the vector fields $Q_1$ and $Q_2$ are $\f$-related.
\end{definition}

Recall that in general for vector fields $Q_1$ and $Q_2$ (that for this purpose do not have to be homological) the condition of being intertwined by a map $\f$ or being  $\f$-related can be formulated in two equivalent ways: either as the commutativity of the diagram
\begin{equation}
    \begin{CD} (\Pi)TM_1 @>{T\f}>> (\Pi)TM_2\\
                @A{Q_1}AA         @AA{Q_2}A\\
                M_1 @>{\f}>>  M_2
    \end{CD}
\end{equation}
(if vector fields are seen as sections of the tangent bundles; if a vector field is even, it is a section of $TN\to M$ and if it is odd, it is a section of $\Pi TM\to M$, so we have or not have $\Pi$ in the above diagram). Or as the equality
\begin{equation}
    Q_1\circ \f^*=\f^*\circ Q_2\,,
\end{equation}
where $\f^*$ is the pullback of functions and vector fields are regarded as operators on functions.

If $x^a$ and $y^i$ are local coordinates on $Q$-manifolds $M_1$ and $M_2$, the condition that a map $\f\co M_1\to M_2$ is a $Q$-morphism is expressed by
\begin{equation}\label{eq.qmorphcoord}
    Q_1^a(x) \der{\f^i}{x^a}(x)= Q_2^i(\f(x))\,,
\end{equation}
where $\f^*(y^i)=\f^i(x)$.

\begin{proposition}
In the examples above, i.e. $\Pi TM$, $\Pi \frakg$ for  a Lie algebra $\frakg$, and the $Q$-manifold corresponding to a complex $(V,{\rm d})$, --- $Q$-morphisms preserving grading are equivalent to, respectively: arbitrary maps $M_1\to M_2$; Lie algebra homomorphisms $\frakg\to \mathfrak{h}$;   chain maps $V\to W$.
\end{proposition}

We see that for maps $\Pi TM_1\to \Pi TM_2$ the condition that a map is a $Q$-morphism is an ``integrability condition''. If we relax preservation of grading, more $Q$-maps appear. For example, general  $Q$-maps $\Pi TM_1\to \Pi TM_2$ in local coordinates are specified by formulas $y^i=\f^i(x,dx)$ (where the r.h.s. is arbitrary), instead of $y^i=\f^i(x)$.

\subsubsection{Zero locus and its involutive distribution.  ``Non-linear homological algebra''}
\begin{definition}
The \emph{zero locus} of a $Q$-manifold $M$ is the zero locus (the set of zeros) of the vector field $Q$. Notation: $\zeroloc(M)$ or $\zeroloc(Q)$.
\end{definition}

In coordinates, if $Q=Q^a(x)\lder{}{x^a}$, then $\zeroloc(M)$ is specified by the equation
\begin{equation}\label{eq.zeroloc}
    Q^a(x)=0\,.
\end{equation}

If $\f\co M_1\to M_2$ is a $Q$-map, it maps $\zeroloc(M_1)$ to $\zeroloc(M_2)$.

In the above examples, we obtain the following subsets as zero loci.

For $\Pi TM$ with $d$, it is specified by the equation $dx^a=0$, hence   $\zeroloc(\Pi TM)=M$.

For a complex $V=(V,{\rm d})$, we obtain  $\zeroloc(V)=\Ker {\rm d}$ (the usual subspace $Z(V,{\rm d})$ of cocycles).

For a Lie algebra $\frakg$, the zero locus  $\zeroloc(\Pi \frakg)\subset \Pi \frakg$ is a conic subspace   given by the quadric equations $c_{ij}^k\xi^i\xi^j=0$.

The zero locus $\zeroloc(M)$ comes equipped with a canonically defined distribution, as follows.
The vector field $Q$ induces a linear transformation $Q_x:=TQ(x)$ in the tangent space $T_xM$ for all $x\in \zeroloc(M)$, and $Q_x^2=0$. It is easy to see (e.g. by using local coordinates) that $\Ker Q_x=T_x\zeroloc(M)$. Hence $\Im Q_x\subset \Ker Q_x$ give a distribution on $\zeroloc(M)$. Denote it $\blocq$, so that $\blocq_x=\Im Q_x$.
\begin{proposition}[\cite{schwarz:semiclassical}, \cite{schwarz:aksz}]
The distribution $\blocq$ on $\zeroloc(M)$ is involutive, $[\blocq,\blocq]\subset \blocq$.
\end{proposition}
\begin{proof}
Observe (e.g. in local coordinates) that the vector fields tangent to $\zeroloc(M)$ can be described as elements of the Lie subalgebra $\Ker(\ad Q)\subset \Vect(M)$ restricted to $\zeroloc(M)$, and the distribution $\blocq$ can be similarly described by the ideal $\Im(\ad Q)\subset \Ker(\ad Q)$, and thus the involutivity follows.
\end{proof}

Hence one may wish to explore the  space of leaves  $\zeroloc(M)/\blocq$, which is a kind of ``non-linear homology''~\cite{schwarz:aksz}. Functions on $\zeroloc(M)/\blocq$ are those functions on $\zeroloc(M)$ that are constant in the directions of $\blocq$.

In the model examples we obtain the following. For a $Q$-manifold corresponding to a complex $(V,d)$, $\zeroloc(M)/\blocq$ coincides with the usual cohomology $Z(V,d)/B(V,d)$. For the rest, the answers are less obvious. One can see that for $\Pi TM$ and for any $x\in M=\zeroloc(\Pi TM)$, $\blocq_x=\Im d_x=\Ker d_x=T_xM$. (There is no homology in the tangent spaces, which is in a certain sense the condition of \emph{non-degeneracy} of a $Q$-structure, see~\cite{schwarz:semiclassical}.) Hence functions on  $\zeroloc(\Pi TM)/\blocq$ are locally constant functions on $M$, i.e. $H^0(M)$, and  $\zeroloc(\Pi TM)/\blocq\cong \pi_0(M)$. This does not feel very satisfying and one may wish to modify the interpretation of $\zeroloc(M)/\blocq$ (e.g. by considering ``points'' that are more general than ordinary $\RR$-valued points, so to be able to detect more information such as the whole cohomology algebra $H^*(M)$).

For  the case of $\Pi \frakg$, one can identify $\zeroloc(\Pi \frakg)/\blocq$ with  the space of orbits of the adjoint action of a Lie group associated with $\frakg$ (note that the adjoint action preserves $\zeroloc(\Pi \frakg)$).
We can elaborate this as follows. It makes sense to consider a slightly more generalize setting.

\begin{example}\label{ex.deffunctor}
Let $\frakg$ be  a \emph{differential Lie superalgebra}, i.e.   besides the Lie bracket it is equipped also with an odd operator $d$ such that $d$ is a derivation of the bracket and $d^2=0$. This is described by a field $Q$ on $\Pi \frakg$ of the form
\begin{equation}
    Q=\left(\x^iQ_i^k + \frac{1}{2}\x^i\x^jQ_{ji}^k \right)\der{}{\x^k}
\end{equation}
(the first term is responsible for $d$, the second for the bracket). In a coordinate-free form,
\begin{equation}
    Q(\x)={\rm d}\x -\frac{1}{2}[\x,\x]
\end{equation}
(the minus sign has some explanation, compare \ref{subsubsec.linfalgebras}).
Hence the equation of the zero locus is
\begin{equation}\label{eq.mc}
    {\rm d}\x -\frac{1}{2}[\x,\x]=0\,.
\end{equation}
Note that points of $\Pi \frakg$ are the same as odd elements of $\frakg$. We look for   solutions of~\eqref{eq.mc} that can depend on arbitrary external parameters  some of which may be odd.
The infinitesimal transformation defined by $Q$ on $\Pi \frakg$ is $\x \mapsto \x +\e Q(\x)$ (where $\e$ is odd). It lifts to the action on arbitrary tangent vectors $\dot \x$ by
\begin{equation}
    \dot\x  \mapsto \dot\x  +\e ({\rm d}\dot \x-[\x,\dot \x])
\end{equation}
(where $\dot \x$ can be of any parity). In particular, if  $\x\in \zeroloc(Q)$,   the action preserves $T_{\x}\Pi \frakg$. So   the linear operator $Q_{\x}$ is the ``covariant derivative'':
\begin{equation}
   Q_{\x}(\h)= {\rm d}\h -[\x,\h]\,,
\end{equation}
where $\h\in \frakg$ (of arbitrary parity).  The equation of the zero locus is  the ``zero curvature'' condition  $({\rm d}-\ad \x)^2=0$.
Hence the tangent space  $T_{\x}\zeroloc(\Pi \frakg)=\Ker Q_{\x}$ consists of  ``covariantly constant'' vectors.
The infinitesimal shift of $\x\in \zeroloc(\Pi \frakg)$, by $\h\in T_{\x}\zeroloc(\Pi \frakg)$, $\x\mapsto \x+\e \h$, in the case of $\h\in \Im Q_{\x}$ is $\x\mapsto \x+\e ({\rm d}\h-[\x,\h])$. This can be viewed as an   ''infinitesimal gauge transformation'' of $\x$, i.e. the infinitesimal form of a transformation $\x\mapsto -{\rm d}g\,g^{-1}+g\x g^{-1}$ by elements of  a  \emph{differential  group}  integrating $\frakg$. This is a usual Lie group with a $Q$-structure coming from ${\rm d}$ on $\frakg$. Hence at least locally the leaves of the distribution $\blocq$ are the same as gauge orbits.
\end{example}

\begin{remark}
Introducing the zero locus of a given homological vector field $Q$ and then taking quotient of it by a distribution on it can be compared with the logic of BRST theory~\cite{henneaux:teitelboimbook} (see also~\cite{lyakhovich-sharapov:2005}, \cite{kazinski-lyakhovich-sharapov:2005}). In BRST theory it   goes in the opposite direction: a given ``constraint surface'' or ``shell'', which has to be factorized by   symmetries generating an involutive distribution, is first ``resolved'' by a version of Tate~\cite{tate:homology1958} or Tyurina~\cite{palamodov:deformations1976} resolutions, which means effectively replacing a submanifold by a non-positively  graded $Q$-manifold which as a fiber bundle over the original ambient manifold, then it is further enlarged  to a $\ZZ$-graded $Q$-manifold (with negative and positive weights) where the vector field $Q$ (denoted traditionally as $s$ and called ``BRST differential'') incorporates   information about constraints  and symmetries. The procedure is non-unique and the correct picture should take care of this non-uniqueness. If one recalls that in standard homological algebra    complexes are taken up to quasi-isomorphism to get the derived category, then  analogously one should expect appearance of some ``derived $Q$-manifolds''. Investigations in this direction coming from the side of  derived algebraic geometry (which has been around for some time)    are already on the way, see e.g. Pridham~\cite{pridham:outline2018}. (See also~\cite{Carchedi:1211.6134},\cite{carchedi-royt:homalg2012}.) The future theory should  to be able to incorporate also microformal morphisms    introduced below in Sections~\ref{sec.microclass},\ref{sec.microquant}.
\end{remark}

\subsubsection{Remark: on deformation of structures using $Q$-manifolds.}
Example~\ref{ex.deffunctor} above was a glance into the apparatus of deformation theory (which will remain  outside the scope of this text). The modern viewpoint is that every algebraic or geometric structure or, better, type of structure,  is controlled by a particular differential graded Lie superalgebra (or its generalization such as an   $\Linf$-algebra,  which we shall  define  in~\ref{subsubsec.linfalgebras}).     Basically, with any such an algebra is associated a \emph{deformation functor} which is roughly  $\zeroloc(M)/\blocq$ for the corresponding graded $Q$-manifold $M$. (Note that we did not consider a $\ZZ$-grading in the example; but in concrete situations it plays important role.) It is roughly a set whose points are moduli or deformations of structures of a considered type. ``Functor'' refers to   dependence on a base of deformations, i.e. a choice of an algebra from which parameters are taken.   The idea that deformations of geometric and algebraic structures are controlled by graded Lie algebras was put forward by Nijenhuis (see~Nijenhuis--Richardson\cite{nijen:deformm}), as an abstract framework modeled on the previous work  on deformations of complex structures (Fr\"{o}licher--Nijenhuis~\cite{nijen:stab}, Kodaira--Spencer~\cite{kodaira-spencer:deform1958} and Kuranishi~\cite{kuranishi:localcomp1962}) and associative algebras (Gerstenhaber~\cite{gerstenhaber:cohomology63}). It was Nijenhuis who brought to the fore the ``deformation equation'' ${\rm d}\x\pm\frac{1}{2}[\x,\x]=0$, called also the Maurer--Cartan equation or master equation. (Nijenhuis was very much ahead of his time, he   possessed for example a working replacement of Lie supergroups under the name ``analytic graded Lie algebras''.)
Then there followed  the work  of Schlessinger--Stasheff~\cite{stasheff:schless} of 1979 and  the work of Goldman--Millson~\cite{goldmanmillson:deftheory88}, who used Deligne's ideas  that quasi-isomorphic DG Lie algebras define the same deformation theory and that instead of taking the deformation functor as a set of equivalence classes, one should   consider it as the corresponding action groupoid (``Deligne's groupoid''). Then it was Kontsevich~\cite{Kontsevich:1997vb} who formulated everything in terms of  formal  $Q$-manifolds, identifying solutions of the deformation equation with points of the zero locus $\zeroloc(Q)$, and  established  invariance of the deformation functor under  $\Linf$  quasi-isomorphisms (much more than   DG Lie!), which was the crucial step for formulating and proving his celebrated formality theorem.

\subsubsection{$Q$-structure on the space of maps.}

This is an analog of the complex of homomorphisms. We will not use this construction in the rest of the paper, but wanted to include it because of its importance. Suppose $M_1$ and $M_2$ are $Q$-manifolds. Consider the infinite-dimensional supermanifold (or graded manifold, if $M_1$ and $M_2$ are graded) of all maps $\Mapp(M_1,M_2)$. Claim: it has a natural $Q$-structure (defined first in~\cite{schwarz:aksz}). It has numerous applications, in original paper~\cite{schwarz:aksz} as well as in many others, e.g.~\cite{cattaneo:sigmaaksz}.

The construction is as follows (we use the exposition given in~\cite{tv:invaksz}).

Consider first arbitrary vector fields $Q_i$ on $M_i$ (not assuming them homological). To them corresponds a vector field on $\Mapp(M_1,M_2)$ that we call the \emph{difference construction}:
\begin{equation}
    d(Q_1,Q_2)[\f]:= \f^*Q_{2}-\f_*Q_1\,,
\end{equation}
where the ``pullback'' $\f^*Q_{2}$ of a vector field $Q_2$ on the target $M_2$ and the ``pushforward'' $\f_*Q_{1}$ of a vector field $Q_1$ on the source $M_1$ are defined respectively as $\f^*Q_{2}:=Q_2\circ \f$ and $\f_*Q_{1}=T\f\circ Q_1$.
Here we treat vector fields as sections of the tangent bundles (not as operators on functions). Both $\f^*Q_{2}$ and $\f_*Q_{1}$ are vector fields along $\f$, i.e. can be perceived as infinitesimal variations of $\f$ or elements of the tangent space $T_{\f}\Mapp(M_1,M_2)$. So is the difference $d(Q_1,Q_2)[\f]$ for each $\f$. Hence we have    vector fields    on  $\Mapp(M_1,M_2)$, in particular, the vector field $d(Q_1,Q_2)$.
The zeros of the vector field $d(Q_1,Q_2)$ are precisely such $\f$ that $Q_1$ and $Q_2$ are $\f$-related.

It is convenient to use the notation  $Q_{2*}$ and $Q_1^*$ for the vector fields   induced on the space of maps, so that $Q_{2*}[\f]=\f^*Q_{2}$  and  $Q_1^*[\f]=\f_*Q_1$. (The position  of the star  corresponds to post- or pre-composition with the infinitesimal diffeomorphism  generated by the vector field.) In this notation,
\begin{equation}
    d(Q_1,Q_2) = Q_{2*}- Q_1^*\,.
\end{equation}
It immediately follows that under  both ``star'' operations, the commutator on $M_1$ or $M_2$ is mapped to th commutator on $\Mapp(M_1,M_2)$, and that any two vector fields   with the lower star and the upper star automatically commute. Hence the main result:

\begin{proposition} For homological vector fields $Q_i$ on $M_i$, the difference construction $d(Q_1,Q_2)$ is a homological vector field on $\Mapp(M_1,M_2)$\,.
\end{proposition}

Explicit formula:
\begin{multline}
   d(Q_1,Q_2) = Q_{2*} - Q_1^*=  \\
   =\int_{M_1}\!\! Dx \, \left(Q_2^i(\f(x)) - Q_1^a(x)\,\der{\f^i}{x^a}(x)\right) \var{}{\f^i(x)}\,.
 \end{multline}
(up to common sign  depending on conventions for the Berezin integral).

For three $Q$-manifolds and a composition of maps $\f_{21}\co M_1\to M_2$ and  $\f_{32}\co M_2\to M_32$ there is a formula~\cite{tv:invaksz}\,:
\begin{multline}\label{eq.composfordqq}
    d(Q_1,Q_3)[\f_{32}\circ \f_{21}]=\\
    =d(Q_2,Q_3)[\f_{32}]\circ \f_{21}+T\f_{32}\circ d(Q_1,Q_2)[\f_{21}]\
\end{multline}
(it is an analog of the Leibniz formula).

\subsection{Digression: derived brackets}
Recall (for reference purposes) the definition of a Lie superalgebra (we prefer not to use ``graded Lie algebras'' to avoid contradiction with the Lie algebras that are graded).

A $\Z$-graded vector space $L=L_0\oplus L_1$ with an even bilinear operation  which we denote  by $[-,-]$ is a \emph{Lie superalgebra} (and the operation is referred to as   `Lie bracket') if antisymmetry
\begin{equation}\label{eq.antisymla}
    [u,v]=-(-1)^{\ut\vt}[v,u]
\end{equation}
and Jacobi identity (which we write in the Leibniz form)
\begin{equation}\label{eq.jacla}
    [u,[v,w]]=[[u,v],w]+(-1)^{\ut\vt}[v,[u,w]]
\end{equation}
are satisfied.

If only~\eqref{eq.jacla} is satisfied (no antisymmetry assumed), then $L$ is called a \emph{Loday} or \emph{Leibniz algebra} and the bracket is referred to as `Loday bracket'.

One can modify these notions by shifting parity so that the bracket becomes odd (with respect to the new parity). Its properties differ by the shift of parities in all the signs. Such structures are called an \emph{odd Lie superalgebra} or an \emph{odd Loday algebra}.

Fix an odd linear operator $D$ on a Loday algebra $L$ which is a derivation of the bracket (for example, $D=\ad \D$ for an odd element $\D$). Define a new   operation of the opposite parity to the original:
\begin{equation}\label{eq.yvetteder}
    [u,v]_D:= \pm [D(u),v]
\end{equation}
(sign not essential and can be properly chosen).

\begin{theorem}[\cite{yvette:derived}]
Suppose $D^2=0$. Then the operation $[u,v]_D$ defines on
$L$ a new Loday algebra structure (of the opposite parity).
\end{theorem}

See also~\cite{yvette:derived2}.    Operation~\eqref{eq.yvetteder} is called \emph{derived bracket}.  It has many   applications. Note that even if the original algebra is a Lie superalgebra, the new algebra does not generally satisfy antisymmetry. However, it may be satisfied (for some elements) for an additional reason.

There is a related construction of ``higher derived brackets'' that we will introduce shortly. They automatically satisfy (anti)symmetry, but at a price that one has to consider an infinite sequence of brackets instead of one.

\subsection{Lie algebroids and multiple Lie algebroids}
\label{subsubsec.lalg}

Recall that a \emph{Lie algebroid} over a manifold $M$ is a vector bundle $E\to M$ with a structure of a Lie (super)algebra on the space of sections and a fiberwise linear map $a\co E\to TM$ over $M$ called anchor, so that
the Leibniz rule is satisfied:
\begin{equation}\label{eq.liealg}
    [u,fv]=a(u)(f) \,v +(-1)^{\ft\ut} f[u,v]
\end{equation}
where $u,v$ are sections and $f$ a function on $M$. (We formulate everything in the super setting.) See~\cite{mackenzie:book2005} as a general source on Lie algebroids and Lie groupoids.

Consider the parity reversed vector bundle $\Pi E\to M$. Let $Q\in \Vect(\Pi E)$ of weight $+1$. If $x^a, \x^i$ are local coordinates on $\Pi E$ so that $\x^i$ of parity $\itt+1$ are linear coordinates on the fibers, the general form of $Q$ is then
\begin{equation}
    Q= \x^iQ_i^a(x)\der{}{x^a}+ \frac{1}{2}\xi^i\x^j Q_{ji}^(x)\der{}{\x^k}\,.
\end{equation}

\begin{theorem}[Vaintrob~\cite{vaintrob:algebroids}]
The structure of a Lie algebroid in $E$ is equivalent to the $Q$-structure on $\Pi E$ of weight $+1$.
\end{theorem}

In other words, if $Q$ as above is odd (this is automatic if $E$ is purely even) and satisfies $Q^2=0$, it defines a Lie algebroid structure in $E$, and conversely. We can give explicit formulas:
\begin{equation}\label{eq.braclalg}
    \iota_{[u,v]}= (-1)^{\ut}[[Q,\iota_{u}],\iota_v]
\end{equation}
and
\begin{equation}\label{eq.anchlalg}
    a(u)(f)=  [Q,\iota_{u}](f)\,.
\end{equation}
Here $\iota_u$ is a vector field on $\Pi E$ of weight $-1$ defined by a section $u\in\fun(M,E)$ by $\iota_u=(-1)^{\ut}u^i(x)\lder{}{\x^i}$ if $u=u^i(x)e_i$. This is an odd isomorphism between $\Vect_{-1}(\Pi E)$ and
$\fun(M,E)$. (Note that there are no vector fields of weights less than $-1$ on $\Pi E$.)

If $E_1$ and $E_2$ are Lie algebroids over the same base $M$, it is not a problem to define a (fixed base) Lie algebroid morphism $E_1\to E_2$. This is just a fiberwise linear map over $M$ preserving brackets and anchors. In particular, $a\co E\to TM$ is itself a Lie algebroid morphism. However, there is no obvious way of defining a Lie algebroid morphism over different bases (because there is no mapping of sections). A highly non-trivial definition was found in~\cite{mackenzie_and_higgins:algebraic}.

\begin{theorem}[\cite{vaintrob:algebroids}]
A fiberwise linear map $E_1\to E_2$ over a map of bases $M_1\to M_2$ is a Lie algebroid morphism if and only if the induced fiberwise linear map $\Pi E_1\to \Pi E_2$ is a $Q$-morphism.
\end{theorem}

This is the most efficient way of dealing with morphisms of Lie algebroids.

Let us mention that a Lie algebroid structure in $E$ is also equivalent to a Poisson bracket on $E^*$ and a Schouten (= odd Poisson, Gerstenhaber) bracket on $\Pi E^*$, both brackets having to be of weights $-1$. This is analogous to the situation for Lie (super)algebras. Later we shall show constructively how all structures on $\Pi E$, $E^*$ and $\Pi E^*$ correspond to each other. (This will be done for the homotopy case, see~\ref{subsubsec.manifestlinf}.)

There is a multiple analog of Lie algebroids: double Lie algebroids, triple Lie algebroids, etc. Double Lie algebroids were first introduced by K.~Mackenzie (see~\cite{mackenzie:crelle}) by using some nontrivial dualization process and then an equivalent simplifying formulation was found in~\cite{tv:qman-mack}.
It can be described as follows. Multiple Lie algebroids live on multiple vector bundles. The simplest way to define a $k$-fold vector bundle is to say that it is a $k$-fold graded manifold (e.g. bi-graded for double vector bundle) such that each of the weights of local coordinates is $0$ or $1$. This leads to a fiber bundle structure with multilinear transition functions (see~\cite{tv:qman-mack}). In particular, a double vector bundle is a commutative square of ordinary vector bundles (plus some extra conditions). Similarly fir the $k$-fold case. There are commuting parity reversions in each of the $k$ directions, and one can consider the total parity reversion. Then a $k$-fold Lie algebroid is specified by $k$ commuting homological vector fields $Q_1$, \ldots, $Q_k$ such that $\w_i(Q_j)=\delta_{ij}$ for the $k$ weights $\w_1, \ldots,\ w_k$. See~\cite{tv:qman-mack}. Double Lie algebroids in particular arise as Drinfeld doubles of Lie bialgebroids introduced in~\cite{mackenzie:bialg}. See~\cite{mackenzie:crelle}.

\subsection{$\Linf$-structure. Higher derived brackets.  ``Non-linear   Lie algebroids''} 

\subsubsection{$L_{\infty}$-algebras}
\label{subsubsec.linfalgebras}

$L_{\infty}$-algebras or ``strongly homotopy Lie algebras'' (SHLA) originated in physics and were mathematically first defined by Lada and Stasheff~\cite{Lada:1992wc}. They exist  in two parallel equivalent versions: ``symmetric'' and ``antisymmetric''. We shall define both. Below we work with  $\Z$-grading only. If necessary, a $\ZZ$-grading can also be taken into account (but it does not affect identities).

\begin{definition}[\emph{$\Linf$-algebra:  antisymmetric version}]
  A  vector space $L=L_0\oplus L_1$ with a collection of multilinear operations called brackets
  \begin{equation}
    [-,\ldots,-]\co \underbrace{L\times \cdots \times L}_{\text{$k$ times}} \to L \quad \text{(for $k=0,1,2,\ldots $)}
    \vspace{-0.5cm}
  \end{equation}
  such that
  \begin{enumerate}[i)]
    \item the parity of the $k$th bracket is $k\mod 2$;
    \item all brackets are antisymmetric (in $\Z$-graded sense);
    \item $\sum_{r+s=n}\sum_{\text{shuffles}}  (-1)^{\b}  [[x_{\s(1)},\ldots,x_{\s(r)}], \ldots , x_{\s(r+s)}]=0$,  for all~$n=0,1,2,3,...$
  \end{enumerate}
  (here $(-1)^{\b}=(-1)^{rs}\sign \s (-1)^{\a}$ and $(-1)^{\a}$  is the Koszul sign).
   \end{definition}

A parallel notion is as follows.

\begin{definition}[\emph{$\Linf$-algebra:  symmetric version}]
  A  vector space $V=V_0\oplus V_1$ with a collection of multilinear operations called brackets
  \begin{equation}
    \{-,\ldots,-\}\co \underbrace{V\times \ldots \times V}_{\text{$k$ times}} \to V \quad \text{(for $k=0,1,2,\ldots $)}
  \end{equation}
  such that
  \begin{enumerate}[i)]
    \item all brackets are odd;
    \item all brackets are symmetric (in $\Z$-graded sense);
    \item $\sum_{r+s=n} \sum_{\text{shuffles}} (-1)^{\a} \{\{v_{\s(1)},\ldots,v_{\s(r)}\}, \ldots , v_{\s(r+s)}\}=0$,  for all~$n=0,1,2,3,...$
  \end{enumerate}
  (here $(-1)^{\a}$ is the Koszul sign).
   \end{definition}
  (Note that here signs  come  from parities only!)

The two variants of an $\Linf$-algebra are related by a change of parity. Let $V=\Pi L$. Then the relation between brackets in $L$ and $V=\Pi L$  is given by the formula
\begin{equation}
 \{\Pi x_1,\ldots, \Pi x_n\}= (-1)^{\e}\Pi [x_1,\ldots,x_n],\,.
\end{equation}
where $\e=\sum \xt_k(n-k)$. Hence it is sufficient to consider just one variant, though in examples both can appear.

It is more convenient to analyze the symmetric version (with all odd brackets). Let $V$ be equipped a symmetric $\Linf$-algebra structure. Because of symmetry, all operations are determined by their values on coinciding even arguments: $\{\x,\ldots,\x\}$ for even $\x\in V$. (We use the letter $\x$ for an even vector in $V$ as a reminder of $V$ being $\Pi L$.) They can be assembled into a formal odd vector field $Q$ on $V$\,:
\begin{equation}\label{eq.assemblq}
    Q(\x)=\sum \frac{1}{n!} \underbrace{\{\x,\ldots,\x\}}_{\text{$n$ times}}\,.
\end{equation}
We can express back the bracket operations in $V$ and $L$ in terms of $Q$, as follows:
\begin{equation}\label{eq.bracklinfodd}
      \{u_1,\ldots,u_n\} =  [\ldots [Q, u_1],\ldots, u_n](0)
\end{equation}
and
\begin{equation}\label{eq.bracklinfeven}
    \iota([x_1,\ldots,x_n])=(-1)^{\e} [\ldots [Q,\iota(x_1)],\ldots,\iota(x_n)](0)\,.
\end{equation}
For elements of $L$, we   use the operation   $\iota$   similar to that   used above for Lie algebroids,
$\iota(x):=(-1)^{\xt}x^i\lder{}{\xi^i}\in \Vect(\Pi L)$ if  $x=x^ie_i\in L$. We denote by $\xi^i$   linear coordinates on $V$ and identify vectors from $V$ with vector fields with constant coefficients.

\begin{theorem} Formulas above define $\Linf$-algebra structures in $V$ and $L$ (in the respective version) if and only if $Q$ is homological, $Q^2=0$.
\end{theorem}

A proof of the theorem follows from a more general construction producing $\Linf$-algebras that we will consider in~\ref{subsubsec.higherder}.

The homological vector field $Q$ has an expansion
\begin{multline}\label{eq.expansq}
    Q=Q^k(\x)\der{}{\x^k}=\\
    \left(Q^k_0 + \x^iQ_i^k +\frac{1}{2}\x^i\x^jQ_{ji}^k+
    \frac{1}{3!}\x^i\x^j\x^{l}Q_{lji}^k+\cdots \right)\der{}{\x^k}\,.
\end{multline}
Up to signs, the Taylor coefficients $Q^k_0$, $Q_i^k$, $Q_{ji}^k$, $Q_{lji}^k$, etc., are structure constants of the $0$-ary, unary, binary, ternary, etc., brackets in $V$ (or $L$). Interpretation of the ``higher Jacobi identities'' in the definition of $\Linf$-algebras is simplified if $Q(0)$ is assumed to be zero. (In general, $Q(0)$ is known as ``curvature'' and the $\Linf$-algebras that we defined are called ``curved''.) Then the first identity says that the unary bracket (which is a linear operator) is a differential; the second identity says that it is a derivation of the binary bracket; the third identity says that the ``usual'' Jacobi identity for the binary bracket is satisfied up to a chain homotopy, the operator of chain homotopy being the ternary bracket. And then there is an infinite sequence of further identities satisfied by the ternary bracket and the ``higher homotopies'' that arise. (This explains ``strongly homotopy'', not just ``homotopy'' in the name.)

\subsubsection{$L_{\infty}$-morphisms}

Here again (after morphisms of Lie algebroids) the superiority of the $Q$-manifold language becomes compelling. Suppose $L$ and $K$ are $\Linf$-algebras in the antisymmetric version, and  $V=\Pi L$ and $W=\Pi K$ are $\Linf$-algebras in the symmetric version. What should be the ``correct'' notion of a morphism? Denote it by a special arrow, $L \rightsquigarrow K$. We have to establish what $L \rightsquigarrow K$ should be.

If we start from a linear map $L\to K$ and require it be a chain map (commute with the differentials), what should be required from it with respect to the binary brackets? It would be too restrictive (and in hindsight, of little use) to require that the binary bracket in $L$ is precisely mapped to the binary bracket in $K$. In view of the homotopy nature of an $\Linf$-structure, it is natural to expect preservation of binary brackets only up to homotopy, which should be considered part of structure. Hence there is an algebraic homotopy operator $\L^2 L\to K$  (equivalently, $S^2V\to W$).

By analogy with the brackets, one expects to have an infinite sequence of such ``higher homotopies'' $\L^kL \to K$ or $S^k(\Pi L)\to \Pi K$ that should be subject to an infinite sequence of identities involving   the higher brackets in $L$ and $K$. Handling such a sequence directly would be very complicated.  It is convenient to turn to the symmetric description. A sequence of linear maps $S^kV\to W$ meant to be ``higher homotopies'' (one can note that they all have to be even) assemble   similarly with what we did for the brackets   into  one formal non-linear map $\f\co V\to W$. (One cannot do the same directly in terms of $L$ and $K$.)

The language of $Q$-manifolds provides now  a one-line solution.

\begin{definition}
An \emph{$\Linf$-morphism} $V \rightsquigarrow W$   is an infinite sequence of even linear maps $S^kV\to W$ which are the Taylor coefficients of a formal non-linear $Q$-morphism  $\f\co V\to W$, where $V$ and $W$ are regarded as formal $Q$-manifolds.

An \emph{$\Linf$-morphism} $L \rightsquigarrow K$ is an infinite sequence of   linear maps $\L^kL\to K$ (of alternating parities; for $k=1$, even) such that the corresponding sequence $S^k(\Pi L)\to \Pi K$ is an $\Linf$-morphism $\Pi L\to \Pi K$.
\end{definition}

Shortly: an $\Linf$-morphism  $V \rightsquigarrow W$ is just a $Q$-mor\-phism $V\to W$; an $\Linf$-morphism  $L \rightsquigarrow K$ is a $Q$-morphism $\Pi L\to \Pi K$. (We do not have to use a special arrow $\rightsquigarrow$ for $V$  and $W$, since it is an ordinary map.)

If $\x^i$ and $\h^{\mu}$ are linear coordinates on $V$ and $W$ respectively, we can write $\f^*(\h^{\mu})=\f^{\mu}(\x)$ and expand as
\begin{equation}
    \f^{\mu}(\x)= \f^{\mu}_0+ \x^i \f^{\mu}_i + \frac{1}{2} \x^i\x^j\f^{\mu}_{ji} + \cdots
\end{equation}
For simplicity assume that $\f^{\mu}_0=0$, i.e. the origin is preserved, and that both algebras have no curvature. Then by expanding the equation of a $Q$-morphism
\begin{equation}
    Q^i_1(\x)\der{\f^{\mu}}{\x^i}= Q^{\mu}_2(\f(\x))
\end{equation}
we obtain, in the first order:
\begin{equation}
    Q_i^j\f_j^{\mu}=\f_i^{\la}Q_{\la}^{\mu}\,,
\end{equation}
and in the second order:
\begin{equation}
    \pm Q_{ij}^k\f_k^{\mu}\pm \f_{i}^{\la}\f_j^{\nu}Q_{\nu\la}^{\mu}= \pm Q_i^k\f_{jk}^{\mu}\pm Q_j^k\f_{ik}^{\mu}\pm \f_{ij}^{\la}Q_{\la}^{\mu}\,.
\end{equation}
up to signs. The first order condition means that the linear term $\f_1\co V\to W$ is a chain map. The second order condition means that $\f_1$ preserves the binary brackets up to a chain homotopy given by $\f_2$. (This is what we have started from  heuristically  above.)
One can obtain in this way the full set of identities that should be satisfied by the Taylor components $\f_k\co S^kV\to W$  (with the proper signs).

We shall give the general formula for even arguments only, hence without signs, but so that the correct signs can be obtained by linearity. (For any multilinear expression, by using auxiliary odd factors, one can make all arguments even and then take  the auxiliary constants out using the linearity, and this would give the desired formula for arguments with arbitrary parities.)

Let $\f\co V\to W$ be a formal map of vector spaces endowed with structures of $\Linf$-algebras. Define its Taylor components (symmetric multilinear maps) by the formulas
\begin{equation}
    \f_{n}(u_1,\ldots,u_n):= \p_{u_1}\ldots \p_{u_n} \f (0)\,,
\end{equation}
where $\p_u$ means the usual derivative along a vector. (We substantially use here the linear structure of $W$, otherwise it would make no invariant sense.) We shall also need the notion of an $\Linf$-structure ``shifted'' by a constant vector. If $\x_0$ is such a vector, we consider a vector field $Q^{\x_0}(\x)=Q(\x+\x_0)$. (Such shifts or ``twistings'' under more abstract guise were considered in~\cite{chuang-lazarev:twistings2011}.) Clearly, if $Q$ is a homological vector field, its shift $Q^{\x_0}$ is again homological vector field. We denote the brackets generated by $Q^{\x_0}$ as $\{-,\ldots,-\}^{\x_0}$. They effectively correspond to expanding $Q$ not at $0$, but at $\x_0$.

\begin{proposition}
The condition that $\f\co V\to W$ is an $\Linf$-morphism is equivalent to  the following infinite sequence of identities, for $n=0,1,2,3, \ldots$ with arguments $u_1, \ldots, u_n\in V$ which   are assumed to be even:
\begin{multline}\label{eq.linfmorphismalg}
    \sum_{k=0}^n \sum_{\substack{\emph{$(n-k,k)$-}\\\emph{shuffles}}}   \f_{k+1}\bigl(\{u_{\s(1)},\ldots,u_{\s(n-k)}\}, u_{\s(n-k+1)},\ldots,u_{\s(n)} \bigr)\\
    =\sum_{r=1}^n\sum_{\substack{i_1+\ldots+i_r=n \\ i_1>0,\ldots, i_r>0}}
     \sum_{\substack{\text{\emph{combinations  $\tau$}}\\\text{\emph{of $i_1,\ldots,i_r$}}\\\text{\emph{out of $n$}}}}
    \bigl\{\f_{i_1}(u_{\tau(1)}, \ldots,u_{\tau(i_1)}),
    \ldots,\\
    \f_{i_r}(u_{\tau(i_{r-1}+1)}, \ldots,u_{\tau(i_r)})\bigr\}^{\f_0}\,.
\end{multline}
\end{proposition}
Here combinations of $i_1,\ldots,i_r$ out of $n$ mean symmetric combinations in each group (order unimportant, e.g. increasing)

For example, we can write down the identities for $n=0,1,2$. \\
For $n=0$:
\begin{equation}
    \f_1(\Omega)=\Omega^{\f_0}\,.
\end{equation}
For $n=1$:
\begin{equation}
    \f_1(\{u\})+\f_2(\Omega,u)=\{\f_1(u)\}^{\f_0}\,.
\end{equation}
For $n=2$
\begin{multline}
    \f_1(\{u_1,u_2\})+  \f_2(\{u_1\},u_2)+ \f_2(\{u_2\},u_1)+  \f_3(\Omega,u_1,u_2)=\\
 =   \{f_2(u_1,u_2)\}^{\f_0}+ \{\f_1(u_1),\f_1(u_2)\}^{\f_0}\,.
\end{multline}
By $\Omega$ we denote  \emph{curvature}, i.e. $\{\varnothing\}$, in any $\Linf$-algebra  Compare with the identities obtain above under the simplifying assumptions that the origin is  fixed and there is no curvature.

\subsubsection{Higher derived brackets} \label{subsubsec.higherder}
We want to explain why the condition $Q^2=0$ for a formal homological vector field on a vector space $V$ encodes the higher Jacobi identities of an $\Linf$-algebra. This can be shown directly, but we will give a more general framework. An abstract setup is as follows. Let $L$ be a Lie superalgebra with a direct sum decomposition into two subalgebras: $L=K\oplus V$. Assume that $V$ is abelian (all brackets are zero). Let $\D$ be an odd element of $L$. Define a sequence of new odd brackets on $V$ by the formula:
\begin{equation}\label{eq.highder}
    \{u_1,\ldots,u_k\}:= P[\ldots [\D,u_1],\ldots,u_k]\,,
\end{equation}
where $P$ is the projection on $V$ parallel to $K$. They are called the \emph{higher derived brackets} generated by $\D$. One can see that they are symmetric (by the Jacobi identity in $L$ and the commutativity of $V$).

\begin{theorem}[\cite{tv:higherder}]
If $\D^2=0$, then the higher derived brackets generated by $\D$ define on $V$ a structure of an $\Linf$-algebra.
\end{theorem}

(There are also   generalizations to arbitrary derivations and a homotopical algebra interpretation, see~\cite{tv:higherderarb}.)

\begin{example}[universal]
Take $L:=\Vect (V)$ for a vector space $V$ regarded as a supermanifold. Then $L=K\oplus V$ where elements of $V$ are treated as constant vector fields and $K$ is the space of vector fields vanishing at the origin. Clearly, these are subalgebras and $V$ is abelian. Projection $P$ is evaluation at zero. Then an arbitrary homological vector field $Q\in \Vect(V)$ defines on $V$ a structure of an $\Linf$-algebra and we arrive at the formulas~\eqref{eq.bracklinfodd}.
\end{example}

This example is universal, i.e. all $\Linf$-algebras arise this way and are specified by some $Q$. However, the advantage of the general construction is that $\Linf$-algebras can also arise from different (not necessarily universal) data $(L=K\oplus V, \D)$. We will meet many examples later.

\subsubsection{$L_{\infty}$-algebroids}
\label{subsubsec.linf-algd}
This notion combines properties of $L_{\infty}$-algebras and Lie algebroids. Let $E\to M$ be a (super) vector bundle.

\begin{de}
An \emph{$L_{\infty}$-algebroid} structure in   $E\to M$  consists of a sequence of brackets  on the space of sections $\fun(M, E)$ defining in it an  (antisymmetric) $L_{\infty}$-algebra structure and a sequence of fiberwise multilinear maps $E\times_M\ldots \times_M E\to TM$ called  anchors,  so that   the Leibniz identities hold:
\begin{multline}
 [u_1,\ldots,u_{n-1},fu_n]=\\
 a(u_1,\ldots,u_{n-1})(f)\,u_n +
(-1)^{\a}f\,[u_1,\ldots,u_n]\,,
\end{multline}
where $(-1)^{\a}= (-1)^{(\ut_1+\ldots +\ut_{n-1}+n)\ft}$.
\end{de}

Consider the parity reversed vector bundle $\Pi E\to M$. We can treat the total space $\Pi E$ as a formal neighborhood of the zero section.

\begin{enumerate}[i)]
  \item  An $L_{\infty}$-algebroid structure on  $E\to M$ is equivalent to a  formal  homological vector field on the supermanifold $\Pi E$.
  \item An  \emph{$L_{\infty}$-morphism} of $L_{\infty}$-algebroids $\Phi\co E_1 \infto E_2$ is   specified by a formal   (in general, nonlinear) $Q$-mor\-phism  $\Phi\co \Pi E_1\to \Pi E_2$.
\end{enumerate}

(We refer to $\Phi\co \Pi E_1\to \Pi E_2$ also as $L_{\infty}$-morphism.)

\begin{example} The collection of all anchors assembles into an $L_{\infty}$-morphism $\Pi E \to \Pi TM$\,.
\end{example}


\subsubsection{``Non-linear Lie algebroids''}
\label{subsubsec.nonlinlalg}
Consider a non-negatively graded manifold $E$. As we know, it is a fiber bundle $E\to M$, where $M=E_0$ (see~\ref{subsubsec.nonegasbundle}) with polynomial transition functions preserving weights. Suppose $N$ is the top weight of local coordinates. (If $N=1$, we come back to vector bundles.) It is possible to develop in such a setup an analog of the Lie algebroid theory~\cite{tv:qman}.

\begin{definition} A structure of a \emph{non-linear Lie algebroid} on a graded manifold $E$ is defined by a (formal) homological vector field $Q\in \Vect(E)$ of weight $+1$.
\end{definition}

What is an algebraic structure associated with such an object? Note that unlike vector bundles,  sections here are not additive, so   not suitable for algebraic operations. As we have found, the ``correct'' vector space is the space of all vector fields of negative weights $\Vectn(E)$. It is a nilpotent (but in general not abelian) Lie subalgebra in $\Vect(E)$. Higher derived brackets can be defined, but are not (anti)symmetric. Because of grading, everything reduces to a differential and a binary derived bracket, on top of the original commutator of vector fields. In~\cite{tv:qman} we have presented a list of identities satisfied by such a structure. Note that one can non-canonically identify $E$ with a graded vector bundle (the normal bundle to the zero section). Then the field $Q$ induces an $\Linf$-algebroid structure in this normal bundle (with an extra grading). It is non-canonical and is defined up to an $\Linf$-isomorphism. The algebraic structure in $\Vectn(E)$ is, on the other hand, canonical. In a sense, both structures contain the same information.

Note that for an non-linear Lie algebroid $E$ there is an anchor $a\co E\to \Pi TM$ defined as the composition $Tp\circ Q$, where $Q$ is regarded as a map $E\to \Pi TE$ and $Tp\co \Pi TE \to \Pi TM$ is tangent to the projection $E\to M$. If $a$ is a fibration, we call  $E$  a \emph{transitive} non-linear Lie algebroid. (This generalizes transitivity for ordinary Lie algebroids~\cite{mackenzie:book2005}.) Note that the anchor is always a $Q$-map, so intertwines $Q$ on $E$ with ${\rm d}$ on $\Pi TM$. Hence for a transitive non-linear Lie algebroid  we can introduce local coordinates as $x^a, {\rm d}x^a, y^i$, where $y^i$ are fiber coordinates over the base $\Pi TM$, and the  homological vector field $Q$ takes the form
\begin{equation}\label{eq.qtransnonlin}
       Q={\rm d}x^a\der{}{x^a}+Q^i(x,dx,y)\der{}{y^i}\,,
\end{equation}
the first term being de Rham differential on $\Pi TM$. This can be compared with the ``$Q$-bundles'' considered by Kotov and Strobl~\cite{strobl:charclasses-2015}, \cite{strobl:talks}. They assumed a bundle structure over $\Pi TM$ with the extra restriction that $Q$ in a local trivialization splits into $d$ and a homological vector field on the standard fiber. Compared with~\eqref{eq.qtransnonlin} this would mean no dependence on $x,dx$ in the second term.  As we showed in~\cite{tv:napl}, the possibility of such a gauge follows from the ``non-abelian Poincar\'{e} lemma''. This covers some part of transitive Lie algebroid theory~\cite{mackenzie:book2005}. An interesting question would be to consider integration of such non-linear Lie algebroids in an analogy with Mackenzie's theory (for transitive Lie algebroids).

\section{Microformal geometry. Classical thick morphisms}
\label{sec.microclass}

In this and the next section, we give a concise introduction to microformal geometry. The key references are:~\cite{tv:nonlinearpullback}, \cite{tv:oscil}, \cite{tv:microformal}   for main ideas and constructions,  also \cite{tv:qumicro}; and \cite{tv:tangmicro} and \cite{tv:highkosz}  for further development and applications.

\subsection{Main constructions}
\subsubsection{Definition of a microformal (thick) morphism}

Let $M_1$, $M_2$ be   supermanifolds  with local coordinates   $x^a$,  $y^i$.
Let $p_a$ and $q_i$ be   the corresponding  conjugate momenta   (fiber coordinates in $T^*M_1$ and $T^*M_2$)
and let  $\omega_1={\rm d}p_a {\rm d}x^a$ and $\omega_2={\rm d}q_i {\rm d}y^i$ be the symplectic forms on   $T^*M_1$ and $T^*M_2$.

\begin{definition}[\cite{tv:nonlinearpullback}, \cite{tv:microformal}]
A \emph{microformal} (aka \emph{thick}) \emph{morphism} $\F\co M_1\tto M_2$ is a formal Lagrangian submanifold $\F\subset T^*M_2\times T^*M_1$ w.r.t.    $\omega_2-\omega_1$ specified locally by a  {generating function} of the form $S(x,q)$\,:
\begin{equation}
    q_i{\rm d}y^i - p_a{\rm d}x^a = {\rm d}(y^iq_i-S)\quad \text{on}\quad \Phi\,,
   \end{equation}
where $S(x,q)$, regarded as   part of the structure, is   a formal power series in the momentum variables on the target manifold $M_2$\,:
\begin{equation}
\begin{aligned}
    S(x,q)&=S_0(x)+ S^i(x)q_i + \frac{1}{2}\,S^{ij}(x)q_jq_i\,+\\
    &\kern.5cm+\, \frac{1}{3!}\,S^{ijk}(x)q_kq_jq_i +\cdots
    \end{aligned}
\end{equation}
We refer to $S$ as the \emph{generating function} of a thick morphism $\F$.
\end{definition}

\begin{remark}
There is   close similarity between our notion of a microformal (thick) morphism between two manifolds and the notion of a symplectic micromorphism between symplectic micromanifolds of~Cattaneo--Dherin--Weinstein~\cite{cattaneo-dherin-weinstein:one}. A ``symplectic micromanifold'' is defined as the germ of a symplectic manifold at a Lagrangian submanifold and a ``symplectic micromorphism'' between such germs is defined as the germ of  a canonical relation between symplectic manifolds representing the germs. Since by the symplectic tubular neighborhood theorem  every symplectic manifold near a Lagrangian submanifold looks like its cotangent bundle, symplectic micromanifolds can be represented by cotangent bundles and every symplectic micromorphism   defines a thick morphism   between the   Lagrangian manifolds  by ``passing from germs to (infinite) jets''. The big difference   is in ``the morphisms of what'' are the corresponding notions. For symplectic micromorphisms, the objects  are    (the germs of) the cotangent bundles. For thick or microformal morphisms, the objects are the manifolds themselves. Hence, we look for an action of such morphisms on functions by an analog of pullbacks by smooth maps. From the viewpoint of symplectic geometry, this would be an action on functions on Lagrangian submanifolds.
\end{remark}

Now we introduce these pullbacks.

\subsubsection{Pullback by a microformal morphism}

Let $\Phi\co M_1\tto M_2$ be a thick morphism with  a generating function $S$.

\begin{definition}
The \emph{pullback} $\Phi^*$ is  a formal mapping $\Phi^*\co \funn(M_2) \to \funn(M_1)$ of functional supermanifolds defined by the formula (see~\cite{tv:nonlinearpullback})
\begin{equation}\label{eq.pull}
    \boxed{ \Phi^*[g] (x)= g(y) + S(x,q) - y^iq_i\,, \vphantom{\der{S}{q_i}}}
\end{equation}
for $g\in \funn(M_2)$, where $q_i$ and $y^i$ are determined from the equations
\begin{equation} \label{eq.pullequations}
   { q_i=\der{g}{y^i}\,(y)\,, \quad y^i=(-1)^{\itt}\,\der{S}{q_i}(x,q)  }
\end{equation}
(giving $y^i=(-1)^{\itt}\der{S}{q_i}(x,\der{g}{y}(y))$ solvable by iterations).
\end{definition}

\begin{remark}
For ordinary manifolds, we do not have to think about parity of functions. (Though we can consider odd functions on purely even manifolds if needed, but they will be families incorporating odd parameters 
rather than individual functions.) For supermanifolds, we have to distinguish between even and odd functions (or `bosonic' and `fermionic' fields in physical parlance) because they satisfy different commutation rules. So above $\funn(M)$ stands for the supermanifold of even (bosonic) functions. Unlike the familiar case, when pullbacks are linear and can be applied to functions regardless of their parity, the pullbacks defined above work only for even functions. (For odd functions, see~\ref{subsubsec.further}.) We use $\funn(M)$ with boldface for a \emph{supermanifold} of even functions (rather than a set) and distinguish it from the $\Z$-graded vector space $\fun(M)=\fun(M)_0\oplus \fun(M)_1$\,.
\end{remark}

Heuristically, if $f=\F^*[g]$, then
\begin{equation}\label{eq.pullascompos}
     \L_f=\L_g\circ \F
\end{equation}
(composition of relations), where $\L_f=\graph(df)$\,. Note that equation~\eqref{eq.pull} contains more information than~\eqref{eq.pullascompos} because~\eqref{eq.pull} is an equality for functions themselves, not   the  derivatives. More important is that~\eqref{eq.pull} and \eqref{eq.pullequations} give  a constructive procedure for calculating pullbacks.

\subsubsection{Description of  pullbacks}
\begin{example} Let $S(x,q)=S^0(x)+ \f^i(x)q_i$. Then: $\F^*[g]= S^0 + \f^*g$ (an ordinary pullback combined with a  shift by a fixed function).
\end{example}

\begin{remark} Ordinary   maps $M_1\to M_2$ can be identified with thick morphism that have generating functions of the form  $S=\f^i(x)q_i$\,, i.e. linear in momenta.
\end{remark}

Write a general generating function as
\begin{equation}
    S(x,q)=S^0(x)+ \f^i(x)q_i+\dots\,
\end{equation}
(note the notation for the linear term).
Then  the equation
\begin{equation}
    y^i=(-1)^{\itt}\der{S}{q_i}(x,\der{g}{y}(y))
\end{equation}
defines a map $\f_g\co M_1\to M_2$ (depending on a function $g$!) as a formal perturbation of the map $\f=\f_0\co M_1\to M_2$ given by the linear term in $S(x,q)$\,:
\begin{equation}
    y^i=\f_g^i(x)=\f^i(x)+ S^{ij}(x)\p_jg(\f(x))+\dots\,,
\end{equation}
and therefore the formula for the pullback becomes
\begin{equation}
    \F^*[g] (x)= \Bigl(g(y) + S(x,q) - y^iq_i\Bigr)\Bigl|_{y=\f_g(x), q=\der{g}{y}(\f_g(x))}\bigr.\,.
\end{equation}
Note that the function $g$ enters in two ways: explicitly as a summand in $g(y) + S(x,q) - y^iq_i$ and implicitly through   $y^i$ and $q_i$. This is the source of the non-linearity (except for the case of a function $S$ linear in the momenta, where the equations for $y$ and $q$ decouple and the dependence on $g$ disappears).

Therefore, for a general thick morphism $\F\co M_1\tto M_2$ the pullback   $\F^*\co \funn(M_2) \to \funn(M_1)$ is    a  formal non-linear differential operator\,:
\begin{equation}
\begin{aligned}
    \F^*[g](x) &= S^0(x) + g(\f(x))\, +\\
    &\kern.5cm+\, \frac{1}{2}\,S^{ij}(x)\p_ig(\f(x))\p_jg(\f(x)) +\dots
    \end{aligned}
\end{equation}
(Higher order terms can also be calculated~\cite{tv:nonlinearpullback}, but their form is not very elucidating.)

As we shall see, these non-linear differential operators possess special properties, so they are far from being arbitrary.

\subsubsection{Coordinate invariance}
\label{subsubsec.coordinvar}
Generating functions of thick morphisms   are not scalar functions, in the sense that they are geometric objects whose representations depend on coordinate systems. They possess the following non-trivial
transformation law.

\begin{transflaw}[for generating functions]
A  generating function $S(x,q)$  as a geometric object on $M_1\times M_2$ transforms by
\begin{equation}\label{eq.transflaw}
    S'(x',q')=S(x,q) - y^iq_i +y^{i'}q_{i'}\,.
    \vspace{-0.2cm}
\end{equation}
Here $S(x,q)$ is the expression for $S$ in `old' coordinates and $S'(x',q')$ is the expression for $S$ in `new' coordinates. At the r.h.s., the variables $x^a$ and $y^{i'}$   are given  by    substitutions: $x^a=x^a(x')$ and  $y^{i'}=y^{i'}(y)$, while
 $q_i$   and    $y^i$   are determined from
\begin{equation}
   q_i=\der{y^{i'}}{y^i}(y)\,q_{i'}\,, \quad y^i=(-1)^{\itt}\der{S}{q_i}(x,q)\,.
\end{equation}
\end{transflaw}

One can see that the cocycle condition is satisfied by this formula (because it has a ``coboundary'' form). This transformation law can either be postulated as part of the definition of thick morphisms or deduced from the requirement that the corresponding formal canonical relation have the same expression in terms of the generating function in all coordinate systems. In all cases, we have the crucial proposition:

\begin{proposition}
If a generating function $S$ transforms according to the transformation law given by~\eqref{eq.transflaw}, the  canonical relation $\Phi\subset T^*M_2\times (-T^*M_1)$ specified by  $S$  and the  operation of pullback $\F^*:\funn(M_2)\to \funn(M_1)$ do not depend on a choice of coordinates.
\end{proposition}

\subsubsection{Key fact: derivative   of     pullback}

As the pullback  by a thick morphism   is a non-linear mapping  of vector spaces of functions (more accurately, we have to speak about the corresponding infinite-dimen\-sional supermanifolds), it is natural to ask about its derivative or variation for a given function $g\in \funn(M_2)$. The answer is remarkable.

\begin{theorem}
\label{thm.tangentmap}
Let $\Phi\co M_1\tto M_2$ be a thick morphism. Consider the pullback
\begin{equation}
    \Phi^*\co \funn(M_2)\to \funn(M_1)\,.
\end{equation}
Then for every  $g\in \funn(M_2)$, the derivative   $T\Phi^*[g]$    is given by
\begin{equation}
    T\Phi^*[g]=\f_g^*\,,
\end{equation}
where $\f_g^*\co \fun(M_2)\to \fun(M_1)$
is the usual pullback with respect to the   map $\f_g\co M_1\to M_2$    defined by
$y^i=(-1)^{\itt}\der{S}{q_i}(x,\der{g}{y}(y))$ (depending perturbatively on $g$, $\f_g=\f_0+\f_1+\f_2 + \ldots$)\,.
\end{theorem}

(Explanation of notation: $\funn(M)$ is the supermanifold of functions, whose   `points' are even functions; $\fun(M)$ is a $\Z$-graded vector space, which can   identified with the tangent space $T_g\funn(M)$ to $\funn(M)$, for an arbitrary $g\in \funn(M)$.)

A direct proof of Theorem~\ref{thm.tangentmap} was given in~\cite{tv:nonlinearpullback}. An alternative proof can be obtained by consideration of quantum thick morphisms (see the next Section~\ref{sec.qman}). (This was suggested by H.~Khudaverdian, whom I thank.)

\begin{corollary}
For every $g$, the derivative $T\F^*[g]$ of $\F^*$ is an  algebra homomorphism $\fun(M_2)\to \fun(M_1)$.
\end{corollary}

\subsubsection{Composition law}
Consider thick morphisms $\Phi_{21}\co M_1\tto M_2$ and\linebreak $\Phi_{31}\co M_2\tto M_3$ with   generating functions $S_{21}=S_{21}(x,q)$ and $S_{32}=S_{32}(y,r)$.  Here $z^{\mu}$ are local coordinates on $M_3$ and by $r_{\mu}$ we denoted the corresponding conjugate momenta.

\begin{theorem} The composition $\Phi_{32}\circ \Phi_{21}$  is well-defined as a thick morphism $\Phi_{31}\co M_1\tto M_3$ with the  generating function $S_{31}=S_{31}(x,r)$, where
\begin{equation}
    S_{31}(x,r)= S_{32}(y,r) + S_{21}(x,q) - y^iq_i
\end{equation}
and $y^i$ and $q_i$ are  expressed through $(x^a, r_{\mu})$  from the system
\begin{equation}
    q_i = \der{S_{32}}{y^i}(y,r)\,, \quad y^i =(-1)^{\itt}\,\der{S_{21}}{q_i}\,(x,q)\,,
\end{equation}
which has a unique solution   as a  power series in $r_{\mu}$ and   a functional power series in $S_{32}$.
\end{theorem}

If we think about thick morphisms as of (formal) canonical relations, there is their composition as ``set-theoretic'' relations. The point of the above statement is that set-theoretic composition leads actually to a relation of the same type, i.e. a thick morphism, and we are able to give a formula for its generating function.

\subsubsection{Further facts}
\label{subsubsec.further}

\paragraph*{(A)  \emph{Formal category}.}

Composition of thick morphisms is associative and   $(\F_{32}\circ \F_{21})^*=\F^*_{21}\circ \F^*_{32}$. In the lowest order, the composition is as in the category $\SMan\rtimes \funn$, whose arrows are pairs $(\f_{21},f_{1})$ with  the  composition   $(\f_{32},f_{2})\circ (\f_{21},f_{1})=(\f_{32}\circ \f_{21}, \f_{21}^*f_{2}+f_{1})$.
Thick morphisms form  a \emph{formal category}  (``formal thickening'' of the category $\SMan\rtimes \funn$).  Notation: $\EThick$.

\paragraph*{(B) \emph{Relation with gradings.}} The notion of a thick morphism as such and the construction of the pullback of functions by a thick morphism do not require   any super or graded structure. We formulated them for supermanifolds with an eye on applications. If necessary, an extra $\ZZ$-grading can be included. One only needs to assume that a generating function $S$ has weight $0$. This would give correct weights for all other quantities in our formulas.

\paragraph*{(C)\emph{``Fermionic version''.}}  There is a    {fermionic  version} based on \underline{anti}cotangent bundles $\Pi T^*M$ and odd generating functions $S(x,y^*)$: ``odd thick morphisms''
\begin{equation}
    \Psi\co M_1\oto M_2
\end{equation}
induce nonlinear pullbacks
\begin{equation}
    \Psi^*\co \pfunn(M_2) \to \pfunn(M_1)
\end{equation}
on \textbf{odd} functions (``fermionic  fields''),  and their composition gives another formal category, $\OThick$, which is a formal thickening of $\SMan\rtimes \pfunn$. What is said above about the possibility of introduction of an extra  $\ZZ$-grading applies in the fermionic case as well.

\subsection{Application to homotopy Poisson brackets}

\subsubsection{$P_{\infty}$- and $S_{\infty}$-structures (homotopy Poisson and Schouten)}

\begin{definition}
A \emph{$P_{\infty}$}-  (resp., \emph{$S_{\infty}$}-) \emph{structure} on   a supermanifold $M$ is an antisymmetric (resp., symmetric) $\Linf$-structure on $\fun(M)$ such that the brackets are multiderivations of the associative product. A supermanifold with a  $P_{\infty}$-structure (resp., an $S_{\infty}$-structure) is called a  \emph{$P_{\infty}$-manifold} (resp., an \emph{$S_{\infty}$-manifold}).
\end{definition}

\begin{enumerate}
  \item A \emph{$P_{\infty}$-structure} on $M$ is specified by an  \textbf{even}  function
  $P\in\fun(\Pi T^*M)$ satisfying $[P,P]=0$, by  the formula
  \begin{equation}\label{eq.pinf}
    \{f_1,\ldots,f_k\}_P :=  [\ldots[P,f_1],\ldots,f_k]|_M\,.
  \end{equation}
  \item An \emph{$S_{\infty}$-structure} on $M$ is specified by an \textbf{odd} function $H\in\fun(T^*M)$ satisfying $(H,H)=0$, by the formula
  \begin{equation} \label{eq.sinf}
    \{f_1,\ldots,f_k\}_H :=  (\ldots(H,f_1),\ldots,f_k)|_M\,.
  \end{equation}
\end{enumerate}
Here   $[-,-]$ stands for the canonical odd Schouten bracket (canonical antibracket) on functions on $\Pi T^*M$ (which can be identified with multivector fields on $M$) and $(-,-)$ stands for the canonical even Poisson bracket on functions on $T^*M$ (i.e. Hamiltonians on $M$). Sometimes we refer uniformly to $H$ or $P$ as to the \emph{master Hamiltonian} of the corresponding $\Sinf$- or $\Pinf$-structure.

It follows that $\Pinf$-brackets have alternating parities: the binary bracket is even, the unary bracket and ternary bracket are odd, etc. A $\Pinf$-structure  on $M$ is a homotopy analog of an ordinary (even) Poisson bracket.

As for $\Sinf$-brackets, they are all odd and an $\Sinf$-structure  on $M$ is a homotopy analog of an odd Poisson (or Schouten or Gerstenhaber) bracket.

Formulas~\eqref{eq.pinf} and \eqref{eq.sinf} are particular cases of ``higher derived brackets''~\cite{tv:higherder,tv:higherderarb}, and the fact that the ``master equations'' $[P,P]=0$ and $(H,H)=0$ imply higher Jacobi identities of $\Linf$-algebras follows from a general theorem from~\cite{tv:higherder}. On the other hand, the universal description of an $\Linf$-algebra is given by a (formal) homological vector field. What are the homological vector fields corresponding to  $\Pinf$- and $\Sinf$-structures?

\begin{theorem}[\cite{tv:nonlinearpullback}]
The homological vector fields corresponding to $\Pinf$- and $\Sinf$-structures with ``master Hamiltonians'' $P\in \fun(\Pi T^*M)$ (even) and $H\in \fun(T^*M)$ (odd) have the Hamilton--Jacobi form:
\begin{equation}\label{eq.qp}
    Q_P=  \int_M Dx\, P\Bigl(x,\der{\psi}{x}\,\Bigr)\var{}{\psi(x)}\in \Vect(\pfunn(M))
\end{equation}
and
\begin{equation}\label{eq.qh}
    Q_H=  \int_M Dx\, H\Bigl(x,\der{f}{x}\,\Bigr)\var{}{f(x)}\in \Vect(\funn(M))\,.
\end{equation}
\end{theorem}

\begin{remark} ``Hamilton--Jacobi'' vector fields such as $Q_H$ live on spaces of functions. They should not be confused with Hamilton vector fields. One can write such a vector field $Q_H$ on $\funn(M)$ (the supermanifold of even functions on $M$)  for any Hamiltonian $H\in\fun(T^*M)$ irrespective of its parity. As it was shown in\cite{tv:nonlinearpullback}, $[Q_{H_1},Q_{H_2}]=Q_{(H_1,H_2)}$ (i.e. the canonical Poisson bracket maps to the commutator of vector fields on $\funn(M)$). The same is true for the fermionic case, i.e. for functions $P$ on $\Pi T^*M$ and vector fields $Q_P$ on $\pfunn(M)$ (the supermanifold of odd functions on $M$).
\end{remark}

\subsubsection{Key theorem: pullback   as an $L_{\infty}$-morphism}
Let $M_1$ and $M_2$ be $S_{\infty}$-manifolds, with   $H_i\in \fun(T^*M_i)$, $i=1,2$.

\begin{definition}[$S_{\infty}$ or ``homotopy Schouten''  thick morphism]
A thick morphism
\begin{equation}
   \Phi\co M_1\tto M_2
\end{equation}
is   \emph{homotopy Schouten} or an $S_{\infty}$ thick morphism   if
\begin{equation}\label{eq.sinfmor}
    \pi_1^*H_1=\pi_2^*H_2\,.
\end{equation}
Here $\pi_i$ are the restrictions on $\Phi$ of the projections of $T^*M_2\times T^*M_1$ on $T^*M_i$.
\end{definition}

Note: condition~\eqref{eq.sinfmor} is expressed by the Hamilton--Jacobi equation for $S(x,q)$
  \begin{equation}\label{eq.schouten}
    H_1\Bigl(x,\der{S}{x}\Bigr)=H_2\Bigl((-1)^q\der{S}{q},q\Bigr)\,. 
\end{equation}

\begin{theorem}  If a thick  morphism of  $S_{\infty}$-manifolds\linebreak $\Phi\co M_1\tto M_2$ is  {$S_{\infty}$}, then the pullback
\begin{equation}
    \Phi^*\co \funn(M_2)\to \funn(M_1)
\end{equation}
is an $L_{\infty}$-morphism of the homotopy Schouten brackets.
\end{theorem}

Explicitly:  if the Hamilton--Jacobi equation~\eqref{eq.schouten} holds, then $\Phi^*$ intertwines the homological vector fields $Q_{H_2}\in \Vect(\funn(M_2))$ and $Q_{H_1}\in \Vect(\funn(M_1))$.

\subsubsection{Analog for $\Pinf$-structures}
Let $M_1$ and $M_2$ be $P_{\infty}$-manifolds, with   $P_i\in \fun(\Pi T^*M_i)$, $i=1,2$.

We have to use the fermionic version of thick morphisms now.

\begin{definition}[$P_{\infty}$ or ``homotopy Poisson'' odd thick morphism]
An odd thick morphism
\begin{equation}
   \Psi\co M_1\oto M_2
\end{equation}
is   \emph{homotopy Poisson} or a  $P_{\infty}$ thick morphism   if
\begin{equation}\label{eq.pinfmor}
    \pi_1^*P_1=\pi_2^*P_2\,.
\end{equation}
Now $\pi_i$ are the restrictions on $\Psi$ of the projections of $\Pi T^*M_2\times \Pi T^*M_1$ on $\Pi T^*M_i$.
\end{definition}

This is expressed by the Hamilton--Jacobi equation for an odd generating function $S(x,y^*)$
  \begin{equation}\label{eq.poisson}
    P_1\Bigl(x,\der{S}{x}\Bigr)=P_2\Bigl(\der{S}{y^*},y^*\Bigr)\,.
\end{equation}

\begin{theorem}
If an odd thick  morphism of  $P_{\infty}$-manifolds $\Phi\co M_1\oto M_2$ is $P_{\infty}$, then the pullback
\begin{equation}
    \Psi^*\co \pfunn(M_2)\to \pfunn(M_1)
\end{equation}
is an $L_{\infty}$-morphism of the homotopy Poisson brackets.
\end{theorem}

That is: if~\eqref{eq.poisson} holds, then $\Psi^*$ intertwines the homological vector fields $Q_{P_2}\in \Vect(\pfunn(M_2))$ and $Q_{P_1}\in \Vect(\pfunn(M_1))$.

Further application that we have in mind (to $\Linf$-algebroids) requires first a short digression. This is another application of the language of thick morphisms, this time having nothing to do in principle with homotopy brackets, but simply to maps of vector spaces or vector bundles. It is as follows.

\subsection{``Non-linear adjoint''}
We shall show that the notion of the adjoint of a linear transformation has an analog for non-linear transformations, but now as a thick morphism rather than an ordinary map. (Again, there are parallel bosonic and fermionic versions.) We work in the setting of vector bundles over a fixed base to avoid complications for different bases. (See~\cite{mackenzie_and_higgins:duality} for duality for vector bundles by using two categories, with ``morphisms'' and ``comorphisms''.)

\subsubsection{The adjoint for a nonlinear transformation} \label{subsubsec.adjoint}

The construction is based on the following fundamental fact.
\begin{thmm}[Mackenzie--Xu~\cite{mackenzie:bialg}]
For dual vector bundles $E$ and $E^*$, there is a diffeomorphism
\begin{equation}
    T^*E\cong T^*E^*\,,
\end{equation}
defined canonically up to some choice of signs. Depending on this choice, it is either symplectomorphism or antisymplectomorphism.
\end{thmm}

The  Mackenzie--Xu diffeomorphism $\kir\co T^*E\to T^*E^*$ plays the key role for our construction of adjoint as a thick morphism.

\begin{theorem}[``adjoint'']
For any  fiberwise, in general nonlinear, map of vector bundles $\Phi\co E_1\to E_2$,
there is   a   {thick morphism}, which we call  the (fiberwise)  \emph{adjoint},
\begin{equation}
    \Phi^*\co E_2^*\tto E_1^*\,.
\end{equation}
The thick morphism $\Phi^*$ is fiberwise in the natural sense. It is an ordinary map and coincides with the usual adjoint  if the map $\Phi$ is fiberwise-linear,  and has the same functorial property $(\F_1\circ \F_2)^*=\F_2^*\circ \F_1^*$.
Construction:
\begin{equation}
    \Phi^*:=\bigl( \kir\times \kir)({\F} \bigr)^{\text{\textrm{op}}}\subset T^*E^*_1\times (-T^*E^*_2)\,,
\end{equation}
where $\kir\co T^*E\to T^*E^*$ is the Mackenzie--Xu diffeomorphism.
\end{theorem}

\begin{corollary}[pushforward]
There is a     pushforward of functions on the dual bundles  (as the pullback by the adjoint)
 \begin{equation}
  \Phi_*:=(\Phi^*)^*\co \funn(E_1^*)\to \funn(E_2^*)
 \end{equation}
that maps  the subspace  of sections    $\funn(M,E_1)\subset \funn(E_1^*)$  to    $\funn(M,E_2)$. It   coincides on sections  with  the obvious pushforward
$\vp\mapsto \Phi\circ \vp$.
\end{corollary}

\subsubsection{The fermionic analog: parity reversed adjoint}

For the fermionic version (for adjoint combined with parity reversion in vector bundles), we need the following analog of the Mackenzie--Xu theorem:

\begin{theorem}[\cite{tv:graded}]
For a vector bundle  $E$, there is a diffeomorphism
\begin{equation}\label{eq.oddmx}
    \Pi T^*E\cong\Pi T^*(\Pi E^*)\,,
\end{equation}
defined canonically up to a choice of signs, and which  is an (anti)symplectomorphism.
\end{theorem}

\begin{theorem}[``antiadjoint'']
For any  fiberwise, in general nonlinear, map of vector bundles $\Phi\co E_1\to E_2$,
there is   an  odd  thick morphism, which we call  the (fiberwise)  \emph{antiadjoint},
\begin{equation}
    \Phi^{*\Pi}\co \Pi E_2^*\oto \Pi E_1^*\,.
\end{equation}
It is an ordinary map and coincides with the usual adjoint combined with parity reversion if the map $\Phi$ is fiberwise-linear. The equality      $(\F_1\circ \F_2)^{*\Pi}=\F_2^{*\Pi}\circ \F_1^{*\Pi}$ holds.
Construction:
\begin{equation}
    \Phi^{*\Pi}:=\bigl( \okir\times \okir)({\F} \bigr)^{\text{op}}\subset \Pi T^*(\Pi E^*_1)\times (-\Pi T^*(\Pi E^*_2))\,,
\end{equation}
where $\okir\co \Pi T^*E\to \Pi T^*(\Pi E^*)$ is the odd analog~\eqref{eq.oddmx} of the Mackenzie--Xu diffeomorphism.
\end{theorem}

\begin{corollary}[pushforward     of functions on the antidual bundles]
For the antidual bundles, there is a     pushforward
 \begin{equation}
  \Phi_*^{\Pi}:=(\Phi^{*\Pi})^*\co \pfunn(\Pi E_1^*)\to \pfunn(\Pi E_2^*)\,.
 \end{equation}
It maps  the subspace  of sections   $\funn(M,E_1)\subset \pfunn(\Pi E_1^*)$  to    $\funn(M,E_2)$. It   coincides on sections  with
$\vp\mapsto \Phi\circ \vp$.
\end{corollary}

\subsection{Application to $\Linf$-algebroids}

In this subsection, we construct a homotopy analog of the familiar relation between Lie algebras and linear Poisson brackets. Recall that a Lie algebra structure for a vector space  (i.e. a Lie bracket defined for its elements) is equivalent to a linear Poisson structure on the dual space (i.e. a Poisson bracket on functions on the dual space), known variably as ``Lie--Poisson'' or ``Berezin--Kirillov'' bracket. Also, a linear map of vector spaces is a Lie algebra homomorphism if and only if its adjoint (which is the map of the dual spaces in the opposite direction) is a Poisson map. In the supercase, to that one can add the similar statements for odd Poisson bracket on the antidual space. What we will do, we will give an analog for $\Linf$-algebroids. This will use $\Pinf$- and $\Sinf$-structures and require  the language of thick morphisms.

\subsubsection{Recollection: manifestations of an   $L_{\infty}$-algebroid structure}
\label{subsubsec.manifestlinf}

Let a (super) vector bundle $E\to M$ have a structure of a $\Linf$-algebroid. Recall from~\ref{subsubsec.linf-algd} that this means a sequence of brackets and a sequence of anchors satisfying certain properties, namely that the brackets define in the space of sections an $\Linf$-algebra structure    and the anchors appear in the Leibniz type formulas for the brackets (with respect to multiplication by functions). We have seen in~\ref{subsubsec.linf-algd}  that such a structure is encoded by a formal homological vector field on the supermanifold $\Pi E$. We can now include into consideration also the bundles $E^*$ and $\Pi E^*$:
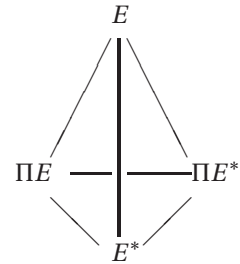
\begin{center}
{ \unitlength=3pt
\begin{picture}(0,35)
\put(0,30)
    {\begin{picture}(0,0)
    \put(0,0){$E$}
    \put(-12,-20){$\Pi E$}
    \put(10,-20){$\Pi E^*$}
    \put(0,-30){$E^*$}
    {\thicklines
    \put(1,-2){\line(0,-1){25}} }
    \put(-5,-19){\line(1,0){5}}
    \put(2,-19){\line(1,0){8}}
    \put(0,-2){\line(-1,-2){7.5}}
    \put(2,-2){\line(1,-2){7.5}}
    \put(-7.5,-22){\line(1,-1){6}}
    \put(10,-22){\line(-1,-1){6}}
    \end{picture}}
\end{picture}
}
\end{center}

\begin{theorem} The following structures are equivalent:
\begin{enumerate}[i)]
  \item $\Linf$-algebroid structure in vector bundle $E\to M$
  \item $\Pinf$-structure on supermanifold $E^*$
  \item $\Sinf$-structure on supermanifold $\Pi E^*$
  \item $Q$-structure (homological vector field) on supermanifold $\Pi L$
\end{enumerate}
(The $\Pinf$- and $\Sinf$-structures must have certain weights, see below.)
\end{theorem}

The quickest way to see that and to obtain explicit formulas is to use the Mackenzie--Xu theorem and its odd analog discussed above.
The cotangent bundle $T^*(\Pi E)$ and the anticotangent bundle $\Pi T^*(\Pi E)$ are double vector bundles:
\begin{equation}
    \begin{CD} T^*(\Pi E) @>>> \Pi E^*\\
                @VVV         @VVV\\
                \Pi E @>>>  M
    \end{CD}
\end{equation}
(and similarly for $\Pi T^*(\Pi E)$) and hence are naturally bi-graded. If   $x^a, \x^i$ are local coordinates on $\Pi E$ (where $\x^i$ are linear coordinates in the fibers), natural coordinates on $ T^*(\Pi E)$ will be $x^a, \x^i, p_a, \pi_i$ (here $p_a, \pi_i$ are the corresponding conjugate momenta). Similarly for $\Pi E^*$: $x^a,\h_i$ and $x^a,\h_i,p_a,\pi^i$. Up to signs, the Mackenzie--Xu transformation is the exchange of $\x^i,\pi_i$ with $\pi^i,\h_i$. The bi-grading is given by the two non-negative weights: $w_1=\#p_a+\#\pi_i=\#p_a+\#\h_i$ and $w_2=\#p_a+\#\x^i=\#p_a+\#\pi^i$. (Incidentally, their difference $w_2-w_1=\#\x^i-\#\pi_i$ is the physicists' ``ghost number'' $\gh$.)

A formal homological vector field
\begin{equation}
    Q=Q^a(x,\x)\der{}{x^a}+ Q^i(x,\x)\der{}{\x^i}
\end{equation}
on $\Pi E$ lifts to an odd Hamiltonian $H\in \fun(T^*(\Pi E))$, where $H=Q\cdot p$, i.e.
\begin{equation}
    H=Q^a(x,\x)p_a+ Q^i(x,\x)\pi_i\,.
\end{equation}
It has weights $w_1(H)=+1$ and $w_2(H)\geq 0$. The application of the Mackenzie--Xu transformation turns $H$ into $H^*\fun(T^*(\Pi E^*))$ of the same weights. Now $w_2$ is the grading by the degrees of the momenta on $\Pi E^*$, so $H^*$ generates an infinite number of odd brackets. Lifting of vector fields maps commutator to the canonical Poisson bracket and the Mackenzie--Xu transformation preserves the Poisson brackets (possibly up to a sign). Hence $[Q,Q]=0$ is equivalent to $(H,H)=0$ and   $(H^*,H^*)=0$. So we get an $\Sinf$-structure on $\Pi E^*$, which is the homotopy analog of the familiar Lie--Schouten bracket for Lie algebras or Lie algebroids. A function $f$ lifted from $\Pi E^*$ to $T^*(\Pi E^*)$ will have weights $w_1(f)=w(f)$ (its degree in $\h_i$) and $w_2(f)=0$. Note the canonical Poisson bracket on $T^*(\Pi E^*)$ has bi-weight $(-1,-1)$. From this we can see that $\Sinf$-brackets on $\Pi E^*$ induced by an $\Linf$-algebroid structure in $E$ will have weights $-n+1$ for an $n$-ary bracket, i.e. $+1$ for the $0$-bracket, $0$ for the $1$-bracket, $-1$ for the $2$-bracket, $-2$ for a $3$-bracket, and so on. (So actually we need to include these weights into the theorem.)

Similar argument applies to $E^*$, where a $\Pinf$-structure, which is  the homotopy  analog of the usual Lie--Poisson bracket, is obtained. Now there is a sequence of brackets of alternating parities, but they again will have  weights $-n+1$ for a bracket with $n$ arguments.

We observe that all these equivalent structures on the three neighbors, $\Pi E$, $\Pi E^*$ and $E^*$ are described basically by one geometric object, which only gets different manifestations.

\subsubsection{$L_{\infty}$-morphisms of Lie--Poisson and Lie--Schouten brackets}


Consider an $L_{\infty}$-morphism $\Phi\co E_1 \infto E_2$ of $L_{\infty}$-algebroids over a base $M$. It is given by a $Q$-map $\Phi\co \Pi E_1 \to \Pi E_2$. To simplify notation, we shall  suppress indications on parity reversion, i.e. write $\F$ instead of $\F^{\Pi}$ and $\F^*$ instead of $\F^{*\Pi}$.

\begin{theorem}
An $L_{\infty}$-morphism $\Phi\co E_1 \infto E_2$ over a base $M$ induces morphisms of the homotopy structures:
\begin{enumerate}[i)]
  \item $\Sinf$ thick morphism $\Phi^*\co \Pi E_2^* \tto \Pi E_2^*$
  \item $\Pinf$ odd thick morphism $\Phi^*\co   E_2^* \oto   E_2^*$
\end{enumerate}
This gives $L_{\infty}$-morphisms of the  homotopy Lie--Schouten  and homotopy Lie--Poisson brackets, respectively (by pushforward):
\begin{equation}
    \Phi_*\co \funn(\Pi E_1^*) \to \funn(\Pi E_2^*)
\end{equation}
and
\begin{equation}
    \Phi_*\co \pfunn(E_1^*) \to \pfunn(E_2^*)\,.
\end{equation}
\end{theorem}

\subsubsection{Example: $L_{\infty}$-morphisms induced by the  anchor}
Recall that the ``higher anchors''
\begin{equation}
    E\times_M\cdots \times_M E\to TM
\end{equation}
for an  $L_{\infty}$-algebroid $E\to M$ assemble into a single (nonlinear, in general) bundle map
\begin{equation}
    a\co \Pi E \to \Pi TM
\end{equation}
over $M$ (which we also refer to as \textbf{anchor}).
\begin{corollary}
\label{coro1}
 The anchor for an $L_{\infty}$-algebroid $E\to M$ induces    $L_{\infty}$-morphisms
 \begin{equation}
   \funn(\Pi E^*)\to \funn(\Pi T^*M)
 \end{equation}
of the homotopy  Schouten brackets,
and
\begin{equation}
   \pfunn(E^*)\to \pfunn(T^*M)\,.
 \end{equation}
of the homotopy  Poisson brackets.
\end{corollary}

\subsubsection{Application to higher Koszul brackets for a homotopy Poisson manifold}
Let $M$ be a $\Pinf$-manifold.
By applying the above to the $\Linf$-algebroid structure induced in $T^*M$ (see~\cite{tv:higherpoisson}) we arrive at the following statement.
\begin{corollary}
On a homotopy Poisson manifold $M$, there is an $L_{\infty}$-morphism
 \begin{equation}
 \boldsymbol\Omega(M)= \funn(\Pi TM)\to \funn(\Pi T^*M)=\boldsymbol{\mathfrak{A}}(M)\,,
 \end{equation}
between the higher Koszul brackets on  forms  induced by the homotopy Poisson structure    and  the canonical Schouten bracket on multivector fields.
\end{corollary}
This gives solution for the problem posed in~\cite{tv:higherpoisson} (where higher Koszul brackets were introduced) and which was the  initial motivation that led us to thick morphisms.
See more details in~\cite{tv:microformal} and~\cite{tv:highkosz}.

\section{Quantum thick morphisms}
\label{sec.microquant}

So far, the statements about bosonic and fermionic thick morphisms were completely parallel to each other. This cannot  remain always   the case because of the substantial difference between even and odd symplectic geometry (see e.g.~\cite{hov:deltabest}, \cite{hov:semi} and \cite{schwarz:bv}; also \cite{tv:laplace1} \cite{tv:laplace2}).
In this section we will see that bosonic thick morphisms governed by even generating functions $S(x,q)$ have quantum counterparts which are special type Fourier integral operators specified by certain ``quantum generating functions'' $S_{\hbar}(x,q)$. In the same way as (classical) bosonic thick morphisms give $\Linf$-morphisms for $\Sinf$-brackets, there is a construction of ``quantum'' $\Sinfh$-brackets (equivalent to a higher order quantum ``BV operator''), which are not $\Sinf$, but tend  to $\Sinf$ when $\hbar\to 0$, and we show how to obtain $\Linf$-morphisms for $\Sinfh$-brackets using quantum thick morphisms. The main references here are~\cite{tv:oscil} and~\cite{tv:microformal}.

\subsection{Main construction}

We treat Planck's constant $\hbar$ as a formal parameter.

\subsubsection{Quantum   pullbacks and quantum  thick morphisms}
We need first to introduce suitable classes of functions. Besides $\fun(M)[[\hbar]]$,   smooth functions on a manifold $M$ which are formal power series in $\hbar$ (``formal power series'' for us always means non-negative  powers), we introduce the algebra of (formal) \emph{oscillatory wave functions}, which we denote $\ofun(M)$, obtained by adjoining to $\fun(M)[[\hbar]]$ formal oscillating exponentials $e^{\frac{i}{\hbar}f(x)}$, where $f\in \fun(M)[[\hbar]]$. The usual rules of manipulating with exponentials are assumed to hold.

\begin{definition} Consider supermanifolds $M_1$ and $M_2$. A \emph{quantum pullback} $\hat\Phi^*$ is a linear operator\linebreak $\hat\Phi^*\co \ofun(M_2)\to \ofun(M_1)$    defined  by the integral formula
\begin{equation}\label{eq.quantpull}
    (\hat\Phi^* [w])(x)=  \int_{T^*M_2} Dy \Dbar q \,\, e^{\frac{i}{\hbar}(S_{\hbar}(x,q)-y^iq_i)}\,w(y)\,,
\end{equation}
where a function  $S_{\hbar}(x,q)$ is called   \emph{quantum generating function}. It is a formal power series in momentum variables on $M_2$\,:
\begin{multline}\label{eq.quantgenfun}
    S_{\hbar}(x,q)=S_{\hbar}^0(x)+
    \f^i_{\hbar}(x)q_i+\\
    + \frac{1}{2}\,S^{ij}_{\hbar}(x)q_jq_i + \frac{1}{3!}\,S^{ijk}_{\hbar}(x)q_kq_jq_i + \ldots
\end{multline}
with coefficients formal power series in $\hbar$.
A   \emph{quantum thick}  (or \emph{microformal}) \emph{morphism}  $\hat\Phi\co M_1\ttoq \, M_2$ is defined as the corresponding arrow in the dual category.
\end{definition}

(In the integral,  $\Dbar q:=(2\pi\hbar)^{-n}(i\hbar)^{m}Dq$  in  dimension   $n|m$.)

There is a question, in which sense to understand oscillatory integrals such as~\eqref{eq.quantpull}. This is achieved by a formal version of the stationary phase formula~\cite{tv:microformal}. An axiomatic theory of formal  oscillatory integrals is developed   by A.~Karabegov in~\cite{karabegov:formal}.

The outward appearance~\eqref{eq.quantgenfun} of a quantum generating function $S_{\hbar}(x,q)$ seems the same as our previous functions $S(x,q)$ apart from a dependence on $\hbar$. We shall see however, that there is a difference: namely, in the respective transformation laws.

\subsubsection{Classical limit}
\begin{theorem}
\label{thm.classlim}
Let $\hat\Phi\co M_1\ttoq \, M_2$ be a quantum thick morphism with a quantum generating function $S_{\hbar}$. Consider $S_0(x,q):=\lim\limits_{\hbar\to 0}S_{\hbar}(x,q)$ as the (classical) generating function  of a (classical) thick morphism $\Phi\co M_1\tto  M_2$.
Then for any   oscillatory wave function of the form  $w(y)=e^{\frac{i}{\hbar}g(y)}$ on $M_2$,
the quantum pullback  is given by
\begin{equation}
\hat\Phi^*\bigl[e^{\frac{i}{\hbar}g}\bigl]= e^{\frac{i}{\hbar}f_{\hbar}(x)}\,,
\end{equation}
where  $f_{\hbar}=\F^*[g](1+O(\hbar))$   and $\F^*$ is the pullback by  the classical microformal morphism $\Phi\co M_1\tto  M_2$   defined by  $S_0(x,q)$.
\end{theorem}
We say that $\F=\lim\limits_{\hbar\to 0}\hat \F$\,.

To be able to   regard legitimately the limit $S_0(x,q)=\lim\limits_{\hbar\to 0}S_{\hbar}(x,q)$ as a classical generating function, we need to know of course that it possesses the required transformation law. We will see that shortly.

\subsubsection{Explicit formula for quantum pullbacks}
Suppose
\begin{equation}
    S_{\hbar}(x,q)=S_{\hbar}^0(x)+\f^i_{\hbar}(x)q_i+S^{+}_{\hbar}(x,q)\,,
\end{equation}
where   $S^{+}_{\hbar}(x,q)$ is the sum of all  terms of order $\geq 2$ in $q_i$.

\begin{theorem}
The action of $\hat \Phi^*$ defined by $S_{\hbar}(x,q)$
can be expressed as follows:
\begin{equation}
     \bigl(\hat \Phi^*w\bigr)(x)=e^{\frac{i}{\hbar}S_{\hbar}^0(x)}
    \left(e^{\frac{i}{\hbar}S^{+}_{\hbar}\left(x,\frac{\hbar}{i}\der{}{y}\right)}w(y)\right)\left|_{y^i=\f^i_{\hbar}(x)}\right.\,.
\end{equation}
\end{theorem}
Hence the quantum pullback $\hat \Phi^*$   is a special type formal linear differential operator   over the `quantum-perturbed' map $\f_{\hbar}\co M_1\to M_2$.
Here   $S_{\hbar}^0(x)$ gives the phase factor,   $\f^i_{\hbar}(x)q_i$  gives the map, and the term $S^{+}_{\hbar}(x,q)$ is responsible for ``quantum corrections''.

\subsubsection{Further facts}

\paragraph*{ (A) \emph{Transformation law.}} Quantum generating functions transform under changes of coordinates by the following formula:
\begin{equation}\label{eq.translawqu}
    e^{\frac{i}{\hbar}S'_{\hbar}(x',q')}=\int Dy \,\Dbar q \, e^{\frac{i}{\hbar}\bigl(S_{\hbar}\left(x(x'),q\right)-yq+y'(y)q'\bigr)}\,.
\end{equation}
Here $x^a=x^a(x')$,  $x^{a'}=x^{a'}(x)$ and  $y^i=y^i(y')$, $y^{i'}=y^{i'}(y)$ are mutually inverse changes of local coordinates on $M_1$ and $M_2$ respectively. In particular, a corollary is that in the limit $\hbar\to 0$, the classical transformation law from~\ref{subsubsec.coordinvar} is recovered. This justifies taking the classical limit in Theorem~\ref{thm.classlim}.

\paragraph*{ (B) \emph{Composition.}} Quantum thick morphisms can be composed. The composition is given by an integral formula similar to that defining quantum pullbacks.


More details see in~\cite{tv:microformal}.

\subsection{Higher BV-structures}

\subsubsection{Digression: brackets generated by an operator}
Let $A$ be   a commutative algebra with $1$ over $\mathbbm{C}[[\hbar]]$.
Let $\D$ be a linear operator on   $A$. Consider two sequences of multilinear operations (of parity $\tilde \D$ and symmetric   in the supersense):
\begin{definition}[a modification of Koszul's~\cite{koszul:crochet85}; see~\cite{tv:higherder}]
\emph{Quantum brackets}  generated by  $\D$\,:
\begin{equation}
    \{a_1,\ldots,a_k\}_{\D,\hbar}:=(-i\hbar)^{-k}[\ldots [\D,a_1],\ldots,a_k](1)\,;
\end{equation}
\emph{classical  brackets} generated by  $\D$\,:
\begin{equation}
    \{a_1,\ldots,a_k\}_{\D,0}:=\lim_{\hbar\to 0}\;(-i\hbar)^{-k}[\ldots [\D,a_1],\ldots,a_k](1)
\end{equation}
\end{definition}

We say that:
\begin{enumerate}[i)]
  \item $\D$ is  a \emph{formal $\hbar$-differential operator} if all quantum brackets are defined;
  \item $\D$ is  an  \emph{$\hbar$-differential operator of order $\leq n$} if all quantum brackets   vanish for $k>n$.
\end{enumerate}

 \subsubsection{More on brackets generated by $\Delta$}

\begin{proposition}[Explicit formulas for quantum brackets] We have:
\begin{enumerate}[i)]
  \item for $k=0$, $\{\varnothing\}_{\D,\hbar} = \D(1)$\,;
  \item for $k=1$, $\{a\}_{\D,\hbar}= (-i\hbar)^{-1}\bigl(\D(a)-\D(1)a\bigr)$\,;
  \item for $k=2$,
   $\{a,b\}_{\D,\hbar}= (-i\hbar)^{-2}\bigl(\D(ab)-\D(a)b-\linebreak(-1)^{\at\bt}\D(b)a+ \D(1)ab\bigr)$\,;
  \item for general $k$,
  the expression for the $k^{\text{th}}$ bracket generated by $\D$ is
\begin{multline*}
    \{a_1,\ldots,a_k\}_{\D,\hbar}=\\
    (-i\hbar)^{-k}
    \sum_{s=0}^k(-1)^s\!\!\!\!\!\!\!\sum_{\text{$(k-s,s)$-shuffles}}\!\!\!\!\!\!\!
    (-1)^{\a}\,\D(a_{\tau(1)}\ldots \\
    { a_{\tau(k-s)})}\,a_{\tau(k-s+1)}\ldots a_{\tau(k)}\,,
\end{multline*}
where  {$(-1)^{\a}= (-1)^{\a(\tau;\at_1,\ldots,\at_k)}$}  is the  Koszul sign   for  permutation of commuting
factors of given parities.
\end{enumerate}
\end{proposition}

\begin{remark}
The notion of an  {$\hbar$-differential operator} can be defined by induction:
 $\ord_{\hbar}\D \leq k$ if  for all $a\in A$,  $[\D,a]=i\hbar B$, where $\ord_{\hbar}B \leq k-1$
 (and $\ord_{\hbar}\D =0$ if $\D$ commutes with multiplication by all $a\in A$).
\end{remark}

\begin{example} \label{ex.hdo}
On a supermanifold $M$, an arbitrary  $\hbar$-differential operator of order $n$ has  in local coordinates the form
\begin{equation}
\begin{aligned}
    &\D=      (-i\hbar)^n  A^{a_1\cdots a_n}_{\hbar}(x)\,\p_{a_1}\cdots\p_{a_n} +  \\
    &\kern.5cm+(-i\hbar)^{n-1} A^{a_1\cdots a_{n-1}}_{\hbar}(x)\,\p_{a_1}\cdots\p_{a_{n-1}}+\cdots + A^0_{\hbar}(x)\,.
\end{aligned}
\end{equation}
For such operators, the \emph{principal symbol} is
\begin{equation}
\begin{aligned}
    \sigma(\D)&=      A^{a_1\cdots a_n}_{0}(x)\,p_{a_1}\cdots p_{a_n}+\\
    &\kern.5cm+ A^{a_1\cdots a_{n-1}}_{0}(x)\,p_{a_1}\cdots p_{a_{n-1}}+     \cdots + A^0_{0}(x)
\end{aligned}
\end{equation}
(subscript $0$ means substituting $\hbar=0$), which is a well-defined inhomogeneous fiberwise polynomial function on $T^*M$.
It is the master Hamiltonian for the classical brackets generated by $\D$.
(For a formal $\hbar$-differential operator, such a master Hamiltonian is a formal power series in momenta.)
\end{example}

 \subsubsection{$S_{\infty,\hbar}$-algebras}
 Let an operator $\D$ on $A$ be odd.  (We assume it is formal $\hbar$-differential.)

\begin{proposition}
If    $\D^2=0$, then  the quantum brackets  define  an $L_{\infty}$-algebra (in the odd symmetric version).
\end{proposition}

This follows from general theory~\cite{tv:higherder}.

The quantum brackets  additionally satisfy the modified Leibniz identity
\begin{multline}
    \{a_1,\ldots,a_{k-1},ab\}_{\D,\hbar} =\{a_1,\ldots,a_{k-1},a\}_{\D,\hbar}b  +\\
    (-1)^{\att}a\{a_1,\ldots,a_{k-1},b\}_{\D,\hbar}
     +\underbrace{(-i\hbar)\{a_1,\ldots,a_{k-1},a,b\}_{\D,\hbar}}_{\text{extra term}}
\end{multline}
where $\att=  {\at(1+\at_1+\cdots+\at_{k-1})}$.

We   call such an algebraic structure an   \emph{$\Sinfh$-algebra}.

Note that an operator $\D$ and the   $\Sinfh$-brackets generated by it contain the same data, and they both
are fully defined by the $0$-bracket and $1$-bracket.

Since an $\Sinfh$-algebra is in particular  an $\Linf$-algebra, we may ask about
the corresponding homological vector field (which should live on $A$). The answer is in the following statement.

\begin{lemma} The quantum  brackets generated by $\D$ correspond to the ``Batalin-Vilkovisky homological vector field''  on $A$ (regarded as a supermanifold) 
\begin{equation}
    Q= e^{-\frac{i}{\hbar}a}\D\bigl(e^{\frac{i}{\hbar}a}\bigr)\,\var{}{a}\,.
\end{equation}
\end{lemma}

 \subsection{BV-manifolds and BV quantum morphisms}

\subsubsection{Definitions}

We introduce the following terminology which may be non-standard, but is convenient for our present purpose.

\begin{definition}
(1) A   \emph{BV-manifold} is a supermanifold $M$ equipped with an odd formal $\hbar$-differential operator  $\D$,    $\D^2=0$.
The operator $\D$ is called the \emph{BV-operator}.

(2) A \emph{(quantum) BV-morphism} of BV-manifolds\linebreak $(M_1, \D_1)$ and $(M_2, \D_2)$ is    a quantum thick morphism\linebreak $\hat\F\co M_1\ttoq M_2$ such that
\begin{equation}
    \D_1\circ \hat \F^* = \hat\F^*\circ \D_2\,.
\end{equation}
\end{definition}

Since a BV-operator $\D$ induces a sequence of quantum brackets,
and is defined by the $0$- and $1$-brackets,   a BV-structure and an $S_{\infty,\hbar}$-structure on a manifold $M$ are  equivalent. In particular, the space of functions on a BV-manifold is an $\Linf$-algebra with respect to quantum brackets generated by $\D$.

A natural question:  how to obtain an $\Linf$-morphism of quantum brackets from a quantum BV-morphism? Note that unlike the classical case, the quantum pullback operator $\hat\F^*$ is linear, so cannot be the answer. It turns out that the solution is given by a formula motivated by the stationary phase method but without passing to the classical limit!

\subsubsection{$L_{\infty}$-morphism of quantum brackets induced by a quantum BV-morphism}
\label{subsubsec.phishriek}
Define a non-linear transformation $\hat \F^!\co \funnh(M_2)\to$\linebreak $ \funnh(M_1)$ by the formula
\begin{equation}
    \hat \F^! :=
     \frac{\hbar}{i} \,\ln\circ\,\hat\Phi^* \circ\exp \frac{i}{\hbar} \,,
\end{equation}
or $\hat \F^!(g)=\frac{\hbar}{i} \,\ln\,\hat\Phi^* \bigl(e^{\frac{i}{\hbar}g}\bigr)$\,, for a $g\in \funnh(M_2)$\,.

\begin{theorem}
If $\hat\F\co M_1\ttoq M_2$ is a BV quantum  morphism, then $\hat \F^!$ is an $L_{\infty}$-morphism of the  $S_{\infty,\hbar}$-algebras of functions.
In greater detail:  $\hat \F^!$
is a morphism of infinite-dimensional $Q$-manifolds $\funnh(M_2)\to \funnh(M_1)$   with the homological vector fields $Q_{\D_1}$ and $Q_{\D_2}$, where
\begin{equation}
    Q_{\Delta}= \int Dx\; e^{-\frac{i}{\hbar}f}\D\bigl(e^{\frac{i}{\hbar}f}\bigr)\,\var{}{f(x)}\,.
\end{equation}
\end{theorem}

Since in the limit $\hbar\to 0$, quantum brackets generated by $\D$ become classical brackets, and the transformation $\F^!$ in the classical limit gives the  pullback $\F^*$ by the corresponding classical thick morphism, as a corollary we obtain that $\F^*$ gives an $\Linf$-morphism of the classical $\Sinf$-brackets. In fact, we can prove more than that.

\subsubsection{From a quantum BV morphism to a classical $S_{\infty}$ thick morphism}
Let $M$    be a BV-manifold with a BV-operator $\D$. In the limit $\hbar\to 0$, $\D$ gives an $\Sinf$-structure.

\begin{lemma}
The  master Hamiltonian of the $\Sinf$-structure generated by $\D$ is
\begin{equation}
    H(x,p)=\lim_{\hbar\to 0} e^{-\frac{i}{\hbar}x^ap_a}\D(e^{\frac{i}{\hbar}x^ap_a})\,.
\end{equation}
\end{lemma}

\begin{theorem}[``analog of Egorov's theorem'']
\label{thm.egorovtype}
Let $M_1$ and $M_2$  be BV-manifolds and let $\hat\F\co M_1\ttoq M_2$ be a BV quantum  thick morphism. Then its classical limit $\F\co M_1\tto M_2$ is an $\Sinf$ thick morphism for the induced $S_{\infty}$-structures.
\end{theorem}

Explicitly: the intertwining relation $\D_1\circ \hat \F^* = \hat\F^*\circ \D_2$ 
implies the Hamilton-Jacobi equation for  
the classical thick morphism $\F=\lim\limits_{\hbar\to 0}\hat\F$\,:
\begin{equation}\label{eq.classthicksinfegor}
    H_1\Bigl(x,\der{S}{x}\,\Bigr)=H_2\Bigl(\der{S}{q},q\Bigr)\,.
\end{equation}

Note: that $\F^*$ is an $\Linf$-morphism of classical brackets follows \emph{per se} from the statement for quantum brackets and $\hat \F^*$. However Theorem~\ref{thm.egorovtype} is a subtler statement: that a  BV  quantum  thick morphism (the intertwining condition for $\D$-operators) induces a classical $\Sinf$ thick morphism (the condition expressed by the Hamilton--Jacobi equation~\eqref{eq.classthicksinfegor}. One can see an analogy with the famous Egorov theorem~\cite{egorov:canonicpdo1971}, which   was one of motivating examples for H\"{o}rmander's theory of Fourier integral operators~\cite{hoer:fio1-1971}. This poses the question about a possibility of quantization for the whole picture: i.e.   lifting of an $\Sinf$-structure to a $S_{\infty,\hbar}$- (= quantum BV) structure and lifting of a classical $\Sinf$ thick morphism to a BV quantum thick morphism. See more in~\cite{tv:microformal}.

\section{Potential further development. Some problems and open questions}
\label{sec.questions}

\subsection{``Non-linear algebra-geometry duality''}
\begin{enumerate}[i)]
\item Define a \emph{non-linear  homomorphism} of (super)al\-gebras to be a non-linear map $A_1\to A_2$ (variant: formal map) such that its derivative at every element $a\in A_1$ is an algebra homomorphism. Question: how to describe such maps?
    \item In particular, is it true that all such non-linear homomorphisms between   algebras $\fun(M)$ are pullbacks by thick morphisms?
\end{enumerate}

 \subsection{Other questions}
\begin{enumerate}[i)]
    \item ``Thick manifolds'': if we consider thick diffeomorphisms, what can be obtained by gluing? Other ''thick`` notions?
    \item Action of thick morphisms on forms, cohomology, etc. ...  (See the  action on tangent bundles~\cite{tv:tangmicro}.)
    \item Analyze analogy between pullbacks by thick morphism with spinor or metaplectic  representation.
    \item Find a characterization of quantum pullbacks among general Fourier integral operators. Find the derivative of the non-linear map $\hat\F^{!}$ (see~\ref{subsubsec.phishriek})
    \item Find a description of quantum and classical pullbacks by generating functions depending on arbitrary variables (as standard in F.I.O. theory). Possibly obtain this way a coordinate-free formulation. (Work in this direction has been very recently initiated by A.~Karabegov.)

    \item Develop more seriously the idea of ``nonlinear homological algebra'' taking graded $Q$-manifolds as a basis and connecting this with the framework of derived algebraic geometry on one hand and practical needs of physics (managing arbitrary choices in BRST formalism) on the other. (Some work relating derived geometry and $Q$-manifolds has started, see e.g.~\cite{pridham:outline2018} and \cite{Carchedi:1211.6134, carchedi-royt:homalg2012}.)
    \item Merge $Q$-manifolds (and the more general picture with derived geometry) with thick morphisms and microformal geometry.
\end{enumerate}

\bibliography{allbibtex}

\bibliographystyle{prop2015}

\end{document}